\documentclass{article}
\usepackage{ijcai25}

\usepackage{times}
\usepackage{soul}
\usepackage{url}
\usepackage[hidelinks]{hyperref}
\usepackage[utf8]{inputenc}
\usepackage[small]{caption}
\usepackage{graphicx}
\usepackage{amsmath}
\usepackage{amsthm}
\usepackage{booktabs}
\usepackage{amssymb}
\usepackage{mathtools}

\usepackage[capitalize,noabbrev]{cleveref}
\usepackage[T1]{fontenc}
\usepackage{mathrsfs}
\usepackage{enumitem}
\usepackage[super]{nth}
\usepackage[toc,page]{appendix}
\usepackage{listings}
\usepackage{color}
\usepackage{tabularx}
\usepackage{csquotes}
\usepackage{multicol}
\usepackage{fancyvrb}
\usepackage{makecell}
\usepackage{dsfont}
\usepackage{thm-restate}

\usepackage{mdframed}
\usepackage[toc,page]{appendix}
\usepackage{mathdots}
\usepackage{xspace}
\usepackage{empheq}
\usepackage{arydshln}
\usepackage{nicefrac}
\usepackage{dsfont}
\usepackage[ruled, linesnumbered]{algorithm2e} 

\usepackage{natbib}
\usepackage{multicol}

\newcommand*{\diff}[1]{\mathop{}\!\mathrm{d}#1} %% https://tex.stackexchange.com/questions/60545/should-i-mathrm-the-d-in-my-integrals

\usepackage{array}
\newcolumntype{L}[1]{>{\raggedright\let\newline\\\arraybackslash\hspace{0pt}}m{#1}}
\newcolumntype{C}[1]{>{\centering\let\newline\\\arraybackslash\hspace{0pt}}m{#1}}
\newcolumntype{R}[1]{>{\raggedleft\let\newline\\\arraybackslash\hspace{0pt}}m{#1}}
\usepackage{thmtools}
\theoremstyle{plain}
\newtheorem{theorem}{Theorem}[section]
\newtheorem{lemma}[theorem]{Lemma}
\newtheorem{corollary}[theorem]{Corollary}

\newtheorem{openq}[theorem]{Open Question}

\newtheorem{definition}[theorem]{Definition}
\newtheorem{example}{Example}

\DeclareMathOperator*{\argmin}{arg\,min}

\DeclareMathOperator{\PAV}{PAV}
\DeclareMathOperator{\CCAV}{CCAV}
\DeclareMathOperator{\EJR}{EJR}
\DeclareMathOperator{\EJRP}{EJR\text{+}}
\DeclareMathOperator{\JR}{JR}

\usepackage[textsize=tiny]{todonotes}

\allowdisplaybreaks

\usetikzlibrary{topaths,calc,math, positioning}

\title{Reconfiguring Proportional Committees}

\author{
Chris Dong$^1$
\and
Fabian Frank$^1$\and
Jannik Peters$^2$\and
Warut Suksompong$^2$\\
\affiliations
$^1$Technical University of Munich\\
$^2$National University of Singapore
}

\begin{document}
\maketitle
\begin{abstract}
An important desideratum in approval-based multiwinner voting is proportionality. We study the problem of reconfiguring proportional committees: given two proportional committees, is there a transition path that consists only of proportional committees, where each transition involves replacing one candidate with another candidate?
We show that the set of committees satisfying the proportionality axiom of \emph{justified representation} (JR) is not always connected, and it is PSPACE-complete to decide whether two such committees are connected. On the other hand, we prove that any two JR committees can be connected by committees satisfying a $2$-approximation of JR. We also obtain similar results for the stronger axiom of \emph{extended justified representation} (EJR). In addition, we demonstrate that the committees produced by several well-known voting rules are connected or at least not isolated, and investigate the reconfiguration problem in restricted preference domains.
\end{abstract}

\section{Introduction}
Multiwinner voting is one of the core topics in computational social choice \citep{FSST17a}, with applications ranging from Q\&A platforms \citep{IsBr24a} to civic participation systems \citep{FGP+23a} to participatory budgeting \citep{PPS21a} and clustering \citep{KePe23a}. 
In particular, a substantial body of literature has focused on \emph{approval-based} multiwinner voting, where voters submit ballots containing the set of candidates that they approve \citep{LaSk23a}.
A common objective in this literature is to identify rules that produce committees which proportionally represent the voters \citep{ABC+16a,PeSk20a, BrPe23a}.

Traditionally, multiwinner voting is a static setting: based on the ballots submitted by voters, a single committee is selected. 
However, this is insufficient for modeling several applications,
especially when a temporal aspect is present.\footnote{We survey related prior work in \Cref{sec:related-work}.\label{fn:related-work}}
Recently, \citet{CHS24a} introduced the \emph{reconfiguration} perspective to multiwinner voting: given an initial committee and a target committee, can we reach the target committee from the initial committee by replacing one candidate at a time? 
For example, these committees could represent catalogs of products displayed by streaming service providers or online stores---changing these catalogs too abruptly might result in a suboptimal customer experience.

If no restrictions are placed on the intermediate committees, the answer to this question is obviously positive, as one can simply replace an excess candidate with a missing one in each step.
The question becomes interesting, however, when the intermediate committees are required to maintain a quality comparable to that of the two given committees. 
\citet{CHS24a} measured ``quality'' using scoring rules, which assign a score to each committee based on the number of approvals of the voters.
Committees that optimize certain scores are known to offer desirable properties; for instance, those maximizing the \emph{Proportional Approval Voting} (PAV) score satisfy a proportionality axiom called \emph{extended justified representation} (EJR) \citep{ABC+16a}. 
Nevertheless, committees with \emph{approximately} optimal scores need \emph{not} satisfy approximate proportionality notions.\footnote{As a simple example, suppose that the voters are split into two halves, where each half unanimously approves $k$ candidates but the two candidate sets are disjoint. 
Selecting a committee with all $k$ candidates from the same half yields a $2$-approximation of PAV. 
However, for any $\alpha\ge 1$, it is possible to choose $k$ and the number of voters so that the committee violates $\alpha$-EJR.}
Thus, the intermediate committees selected by Chen \emph{et al.}~do not necessarily retain the proportionality guarantees of the initial and target committees.
This inspires us to study the following problem: 
\begin{quote}
\emph{Given two committees $W$ and $W'$ satisfying a proportionality axiom, does there exist a transition from $W$ to $W'$ such that each step in the transition involves replacing a candidate by another candidate, and every intermediate committee also (approximately) satisfies this proportionality axiom?}  
\end{quote}
Such problems are typical in the field of reconfiguration \citep{Nish18a},\footref{fn:related-work} and we will examine them for different proportionality notions in approval-based multiwinner voting.
Aside from their intrinsic interest, investigating these problems will also yield insights into the structure of proportional committees, for example, whether the space of such committees is always connected or whether a proportional committee can be isolated in this space.

\subsection{Our Results}

In \Cref{sec:JR}, we begin by examining the set of committees satisfying the axiom of \emph{justified representation} (JR).
We provide an example of JR committees that are not connected by any transition path containing only JR committees. 
In fact, these disconnected JR committees can be heavily isolated: it is possible that for a JR committee, no other JR committee is within distance $k-2$ of it (and this bound is tight), where $k$ denotes the size of the committees.
We also prove that deciding whether two JR committees are connected is PSPACE-complete.
On the other hand, by relaxing the proportionality requirement for the intermediate committees, we show that any two JR committees can be connected via a path containing only ``$2$-approximate'' JR committees, and the factor of $2$ cannot be improved.
Furthermore, we examine the stronger axiom of \emph{extended justified representation} (EJR). 
While the existence of isolated JR committees transfers to EJR, any two EJR committees are connected via a path containing only $4$-approximate EJR committees. 

Next, in \Cref{sec:specific-rules}, we turn to proportional voting rules and present a general positive result: we show that the committees selected by several popular rules---including PAV, the Method of Equal Shares, and sequential-Phragmén---are all reachable from one another via committees satisfying JR. 
To prove this result, we show that each committee selected by these rules is connected to some committee containing an ``affordable'' JR subcommittee, and that all such affordable subcommittees are connected among themselves.
This result stands in stark contrast to our earlier finding that there exist committees which are heavily isolated despite satisfying  JR. 

Finally, in \Cref{sec:restricted-domains}, we study restricted preference domains and show that for the well-studied \emph{voter interval} (VI) and \emph{candidate interval} (CI) domains, the set of JR committees is always connected, that is, any JR committee can reach any other JR committee through a path containing only JR committees. 
In fact, on the CI domain we can always transition between two JR committees without involving auxiliary candidates that do not belong to the initial or target committee. 
By contrast, on the VI domain, this is true only if the committees contain no ``Pareto-dominated'' candidates. 

\subsection{Related Work}
\label{sec:related-work}

Besides the work of \citet{CHS24a} that we discussed earlier, our work is related to recent papers considering dynamic and online aspects of multiwinner voting.
\citet{DoPe24a} studied a model in which the candidate set changes dynamically over time, and each chosen committee should satisfy proportionality without differing too much from the previous committee.
A major difference between our work and theirs is that the set of candidates is static in our model, and the initial and target committees are given.
\citet{DHLS22a} investigated an online setting where candidates arrive over time and each candidate must be irrevocably chosen or rejected.
Several authors examined models that involve selecting a series of committees \citep{BKN20a,BFK22a,DFB23a,ZBET24a}, with some allowing for changes in voter preferences over time; however, these authors did not assume predetermined initial or target committees and also did not consider proportionality.
For a broader overview of research on temporal multiwinner voting, we refer to the survey by \citet{EOT24a}.

Reconfiguration has long been studied for various algorithmic problems such as minimum spanning tree \citep{IDH+11a} and graph coloring \citep{JKK+16a}, and recently gained increasing attention in social choice theory.
\citet{IKSY24a} and \citet{CNR25a} explored it in the context of fair division, whereas \citet{IIK+22a} addressed a similar problem on envy-free matchings.
\citet{OEFS13a,OEF20a} considered the connectivity and convexity of elections that select the same winner under certain voting rules, where two elections are considered adjacent if they can be obtained from each other by swapping adjacent candidates in the ranking of a single voter.
In addition, reconfiguration is related to the framework of ``one profile, successive solutions'' proposed by \citet{BoNi21a} for computational social choice.

\section{Preliminaries and Notation}
\label{sec:prelims}

In this section, we recap the standard setting of approval-based multiwinner voting, discuss a selection of known rules, and introduce terminology for the connectedness of (sets of) committees.
For any positive integer $t$, let $[t] = \{1,2,\dots,t\}$.

\paragraph{Approval-Based Multiwinner Voting}
We assume that we are given a set of $n$ \emph{voters} $N = \{v_1,\dots,v_n\}$
and a set of $m$ \emph{candidates} $C = \{c_1, \dots, c_m\}$. 
An \emph{approval profile} $A = (A_v)_{v\in N}$ is a collection of \emph{approval ballots} $A_v\subseteq C$, where each voter $v \in N$ submits the set of candidates that she approves. When a voter $v_i$ is given with index, we write $A_i$ instead of $A_{v_i}$ for convenience. Further, $N_{c} \coloneqq \{v\in N\mid c\in A_v\}$ denotes the \textit{support} of a candidate $c\in C$, and $N_W \coloneqq \{v\in N\mid W\cap A_v\neq \emptyset\} $ denotes the support of  a subset of candidates $W\subseteq C$.
Given a \emph{target committee size} $k\le m$, we call a subset of candidates $W\subseteq C$ a \emph{committee} if $\lvert W \rvert = k$ and a \emph{subcommittee} if $\lvert W \rvert \le k$.
An \emph{instance} consists of an approval profile $A$ and a committee size $k$.
A \emph{voting rule} $f$ takes as input an instance $(A,k)$ and outputs a non-empty set of committees $f(A,k)$. 

In this paper, we study rules and committees that proportionally represent the electorate. To define proportionality, we follow the \emph{justified representation} approach initiated by \citet{ABC+16a}. For $\ell \ge 0$, we say that a group $N'\subseteq N$  of voters is \emph{$\ell$-large} if $\lvert N'\rvert \ge \ell \cdot\frac nk$. We further say that $N'$ is $\ell$-\emph{cohesive} if $N'$ is $\ell$-large and $\lvert \bigcap_{v\in N'} A_v\rvert \ge \ell$, i.e., these voters approve at least $\ell$ candidates in common.

\paragraph{Proportionality Axioms} A (sub)committee $W$ satisfies \emph{justified representation} (JR) if for every $1$-cohesive group of voters $N'$, there exists a voter $v\in N'$ with $A_v\cap W \neq \emptyset$. In other words, for every group of voters of size at least $\frac{n}{k}$ who approve at least one common candidate, at least one of these voters must approve some candidate in the (sub)committee.
Similarly, a (sub)committee $W$ satisfies \emph{extended justified representation} (EJR) if for every $\ell \in \mathbb{N}$ and every $\ell$-cohesive group of voters $N'$, there exists a voter $v\in N'$ with $\lvert A_v\cap W \rvert \ge \ell$.
Further, a (sub)committee $W$ satisfies the strengthening \emph{EJR+} if for every $\ell \in \mathbb{N}$ and every $\ell$-large group of voters $N'$, either $\bigcap_{v\in N'} A_v \subseteq W$ or there exists a voter $v\in N'$ with $\lvert A_v\cap W \rvert \ge \ell$. 
We will also study approximate versions of these proportionality axioms. For $\alpha \ge 1$, we say that a (sub)committee satisfies $\alpha$-JR (resp., $\alpha$-EJR) if for every $\alpha$-cohesive group of voters $N'$ (resp., $(\alpha \ell)$-cohesive group), there exists $v\in N'$ such that $\lvert A_v\cap W\rvert \geq 1$ (resp., $\lvert A_v\cap W\rvert \geq \ell$ for each $\ell\in\mathbb{N}$). 
Observe that for $\alpha = 1$, this is equivalent to JR and EJR, respectively.
A \textit{witness} of a proportionality violation is a tuple of candidate, voter group, and parameter $\ell\in [k]$ such that the respective proportionality notion is violated; sometimes, we refer to only a candidate as a witness.  
For example, if there is a candidate $c\notin W$ and an $\ell$-large group of voters $N' \subseteq N_c$ such that $\lvert A_v\cap W \rvert < \ell$ for all $v\in N'$, then $(c,N',\ell)$ (or just $c$) is called a witness of an EJR+ violation for $W$. 
For JR, which only considers $\ell =1$, we can omit the parameter $\ell$, whereas for EJR, any candidate in the set of $\ell$ candidates is a witness.
Throughout the paper, we let $\JR(A,k)$ (resp., $\EJR(A,k)$, $\EJRP(A,k)$) denote the set of all committees satisfying JR (resp., EJR, EJR+) for a given instance. 
Similarly, we define $\alpha$-$\JR(A,k)$ and $\alpha$-$\EJR(A,k)$ for any $\alpha \ge 1$.

We say that a voting rule satisfies a given axiom if for every instance $(A,k)$ and every $W\in f(A,k)$, the committee $W$ satisfies the axiom for this instance. 

\paragraph{Common Rules} Next, we recap approval-based multiwinner voting rules that are known to satisfy JR or EJR(+), and introduce the relevant terminology.
For more detailed discussions, we refer to the comprehensive book by \citet{LaSk23a}.
Note that according to our definitions of GJCR, MES, and GreedyEJR, it can happen that these rules do not return committees but only subcommittees---in this case, we extend the subcommittees to committees arbitrarily.

\underline{Greedy Justified Candidate Rule (GJCR):}
This rule iteratively detects the largest current EJR+ violation and adds the responsible candidate to the committee.
Under GJCR, the addition of a candidate incurs a cost of $1$, which is uniformly split among all voters involved in the violation.
For a formal definition, see Algorithm \ref{alg:gjcr_w_prices}. 

\begin{algorithm}[tbh]
\caption{GJCR by \citet{BrPe23a}}
\label{alg:gjcr_w_prices}

 $W \gets \emptyset$\;
 $p(v) \gets 0$ for all $v\in N$\;
\For{$\ell \in \{k, k-1, \dots, 1\}$}{
$N(c) \gets \{v \in N_c \mid \lvert A_v \cap W \rvert < \ell\}$ for $c\notin W$\; \label{line:GJCR-violation}
\While{there is $c\notin W$ such that $\lvert N(c) \rvert \ge \ell \cdot\frac{n}{k}$ \label{line:GJCR-large-set} }{
 Add candidate $c$ maximizing $\lvert N(c)\rvert$ to $W$\;
 $p(v)\gets p(v) + \frac{1}{\lvert N(c)\rvert}$ for all $v\in N(c)$\;
}
}
 return $W$\;
\end{algorithm}

\citet{BrPe23a} have shown that each voter spends at most $p(v)\le \frac kn$ in total, so any (sub)committee returned by GJCR contains at most $k$ candidates.

\underline{Method of Equal Shares (MES):}
This rule starts with a budget of $\frac kn$ per voter and an empty committee. Adding a candidate to the committee incurs a cost of $1$. In each step, MES selects a candidate that minimizes the maximum amount of budget any voter has to pay. A detailed description can be found in Algorithm \ref{alg:mes}.

\begin{algorithm}[tbh]
\caption{MES by \citet{PeSk20a}}
\label{alg:mes}

 $W \gets \emptyset$\;
 $b(v) \gets \frac kn$ for all $v\in N$\;
 $D\gets \{c\in C\setminus W\mid \sum_{v\in N_c} b(v) \ge 1\}$\;
 
\While{$D\neq \emptyset$}{
    \For{$c\in D$}{
        $\widehat{q}(c)\gets \min \{q \ge 0 \mid  \sum_{v\in N_c} \min \{b(v),q \} \ge 1\}$\;
    }
    $c^*\gets \arg\min_{c\in D} \widehat{q}(c)$\;
    $W\gets W\cup \{c^*\}$\;
    \For{$v\in N_{c^*}$}{
        $b(v)\gets b(v) - \min \{b(v),\widehat{q}(c^*)\}$\;
    }
    $D\gets \{c\in C\setminus W\mid \sum_{v\in N_c} b(v) \ge 1\}$\;
}
 return $W$\;
\end{algorithm}

We will use the following notation for output committees $W$ of both MES and GJCR. For $v\in N$, we denote the set of candidates that $v$ pays for as $W(v)$, and for each $c\in W$, we denote the set of voters who pay for $v$ as $N(c)$. Note that $c\in W(v)$ if and only if $v\in N(c)$. 

\underline{Proportional Approval Voting (PAV):}
The PAV score of a subcommittee $W$ with respect to a profile $A$ is defined as $\mathrm{PAV}(W) \coloneqq \sum_{i \in N} H(\lvert A_i\cap W\rvert)$,
where $H(x) \coloneqq \sum_{y \in [x]} \frac 1y$ is the $x$-th harmonic number.
Given a profile $A$, 
the PAV rule outputs all committees 
maximizing the PAV score. 
When studying this rule, we will sometimes consider the marginal contribution of a candidate $c$ with respect to $W$, which is defined as $\Delta(W,c) \coloneqq \PAV(W) - \PAV(W \setminus \{c\})$.

\underline{CCAV and seqCCAV:}
The Chamberlin--Courant score of a committee is $\mathrm{CC}(W) \coloneqq \lvert \{v \in N \mid A_v \cap W \neq \emptyset\}\rvert$. The Chamberlin--Courant rule (CCAV) selects the committees $W$ maximizing $\mathrm{CC}(W)$. The sequential-CCAV rule (seqCCAV) starts with an empty subcommittee $W$ and iteratively adds an unselected candidate $c \notin W$ that maximizes $\mathrm{CC}(W \cup \{c\})$ to $W$.

\underline{GreedyEJR:}
The GreedyEJR rule starts with an empty subcommittee $W$. 
While $W$ does not satisfy EJR, it takes the largest $\ell \le k$ such that there exists an $\ell$-cohesive group $N' \subseteq N$ of voters, adds $\ell$ candidates from $\bigcap_{v \in N'} A_v$ to $W$, and deletes 
the voters in $N'$ from the instance.

\underline{seqPhragmén:}
The seqPhragmén rule begins by assigning each voter $v \in N$ a budget $b(v)$, which starts at $0$. 
It then increases these budgets at the same speed for all voters; assume without loss of generality that each voter receives one unit of budget per one unit of time. 
Whenever an unselected candidate $c$ has $\sum_{v \in N_c} b(v) = 1$, this candidate is added to the committee, and $b(v)$ is reset to $0$ for all $v \in N_c$.

GJCR, MES, and PAV satisfy EJR+, GreedyEJR satisfies EJR but not EJR+, while CCAV, seqCCAV, and seqPhragmén satisfy JR but not EJR \citep{BrPe23a,BrPe24a,LaSk23a}.

\paragraph{Distance, Connectedness, and Isolation} To measure the distance between two committees $W,W'$, we consider the size of the symmetric difference between them and define $d(W,W') = \lvert W\setminus W'\rvert = \lvert W'\setminus W\rvert$. 

Two committees $W,W'$ are called \textit{connected} in some set $\overline {\mathcal S}$ if there exists a sequence of committees $W_1,\dots, W_x \in \overline{\mathcal S}$ with $W_1 = W$ and $W_x = W'$ such that $d(W_{i}, W_{i+1}) = 1$ for all $i\in[x-1]$. 
Two subcommittees 
are called \textit{connected} in $\overline {\mathcal S}$ if all pairs of committees containing them are connected in $\overline {\mathcal S}$.
A set of subcommittees $\mathcal S$ is called \emph{connected} in a set $\overline{\mathcal S}$ if each pair of subcommittees $W,W'\in \mathcal S$ is connected in $\overline{\mathcal S}$.
Often, when $\overline{\mathcal S} = \mathcal S$, we simply say that $\mathcal S$ is 
connected.

In contrast, for $r\ge 1$, a committee $W$ is called $r$-\emph{isolated} in  $\mathcal S$ if $\lvert \mathcal S\rvert \ge 2$,  $W\in \mathcal S$, and for all committees $W'\ne W$ with $d(W,W') \leq r$ it holds that $W'\notin \mathcal S$. 
If a committee is $1$-isolated, we call it \emph{isolated} for short.

As an example, let $A=\{c_1,c_2,c_3,c_4\}$ and $k= 2$. \Cref{fig:committee_graph}
 illustrates all committees as nodes of a graph, with an edge between two committees if the distance between them is one. The committee $\{c_1,c_2\}$ is ($1$-)isolated in $\mathcal S = \{ \{ c_1,c_2\}, \{ c_3,c_4\}\}$, while the set $\{ \{ c_1,c_2\}, \{ c_3,c_4\}, \{ c_1,c_4\}\}$ is 
 connected.

\begin{figure}
    \centering
    \begin{tikzpicture}[scale=1.5, every node/.style={circle, draw, fill=white, minimum size=8mm, inner sep=0}]
    
    \node (1) at (90:1) {\small $\{c_1,c_2\}$};
    \node (2) at (30:1) {\small $\{c_1,c_3\}$};
    \node (3) at (330:1) {\small $\{c_2,c_3\}$};
    \node (4) at (270:1) {\small $\{c_3,c_4\}$};
    \node (5) at (210:1) {\small $\{c_2,c_4\}$};
    \node (6) at (150:1) {\small $\{c_1,c_4\}$};

    \draw (1) -- (2);
    \draw (2) -- (3);
    \draw (3) -- (4);
    \draw (4) -- (5);
    \draw (5) -- (6);
    \draw (6) -- (1);
    \draw (1) -- (3);
    \draw (1) -- (5);
    \draw (2) -- (4);
    \draw (2) -- (6);
    \draw (3) -- (5);
    \draw (4) -- (6);
\end{tikzpicture}
\caption{Each node represents a committee of size two. An edge is drawn between two committees if the distance between them is one.}
\label{fig:committee_graph}
\end{figure}
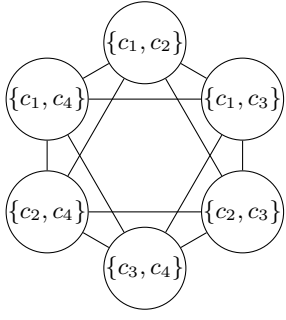

\begin{example}
    Consider the following instance with $n = 9 $ and $k = 3$ (so $\frac{n}{k} = 3$), and the approval profile
    \begin{align*}
        1 \times  \{c_1\},\,\,
        2 \times \{c_1,c_2\},\,\,
        4 \times \{c_3,c_4,c_5\},\,\,
        2 \times  \{c_3,c_4,c_6\}.
    \end{align*}
    Two possible JR committees are $\{c_1,c_3,c_4\}$ and $\{c_2,c_5,c_6\}$. The former committee additionally satisfies EJR, while the latter committee does not---the violation is caused by the last six voters, who would be underrepresented. 
    The two committees are connected in the set of JR committees, e.g., via the following path:
    \[
    \{c_1,c_3,c_4\} \to \{c_2,c_3,c_4\} \to \{c_2,c_5,c_4\} \to \{c_2,c_5,c_6\}.
    \]
    However,
     \[
    \{c_1,c_3,c_4\} \to \{c_5,c_3,c_4\} \to \{c_2,c_5,c_4\} \to \{c_2,c_5,c_6\}.
    \]
    is not a valid path, as $\{c_5,c_3,c_4\}$ does not satisfy JR.
\end{example}

\section{(Extended) Justified Representation}
\label{sec:JR}

Our definitions from \Cref{sec:prelims} immediately lead us to our first question: is the set $\JR(A,k)$ of committees satisfying JR always connected? 
We provide a strong negative answer to this question.
In particular, we show that there exists a JR committee with the property that no other committee of distance at most $k-2$ from it satisfies JR. 
\begin{theorem}
For any $k \ge 3$, there exists an approval profile $A$ and a committee $W \in \JR(A,k)$ that is $(k-2)$-isolated in $\JR(A,k)$, that is, every other committee 
in $\JR(A,k)$
has distance at least $k-1$ to $W$.
\label{thm:isolated_jr}
\end{theorem}
\begin{proof}
    Let the number of voters be $n = k^3$. We shall construct a profile and specify a committee $W$ satisfying JR, such that all other committees satisfying JR have a distance of at least $k-1$ to $W$ and, moreover, $W$ is not the only committee satisfying JR in this instance.

    Let $N = N_1 \cup N_2$ be the set of voters partitioned into two groups, $N_1 = \{v_1, \dots, v_{k^2}\}$ and $N_2 = \{v_{k^2+1}, \dots, v_n\}$. Furthermore, let $C = C_1 \cup C_2$ be the set of candidates, which we partition into a weak set ($C_1$) and a strong set $(C_2)$.
    We choose $C_1 = \{c_1,\dots, c_{k} \}$ such that $c_i$ is approved by consecutive voters $v_{ 1+ (i-1) k }$, $
    v_{2 + (i-1) k }$, $ \dots, v_{k + (i-1)k }$. Note that this covers all the voters in $N_1$ and means that each candidate in $C_1$ is approved by exactly $k$ voters.
    
    Additionally, for every $i\in [k^2]$ and every subset $N_2'\subseteq N_2$ with $\lvert N_2'\rvert = k^2-1$, we create a candidate $d_{i,N_2'}$  approved by the single voter $v_i$ and every voter in $N_2'$. Formally, we set $C_2 =\{d_{i,N_2'} \mid i\in [k^2],\, N_2'\in \binom{N_2}{k^2-1}\}$. Hence, every candidate in $C_2$ is approved by exactly $k^2 = \frac{n}{k}$ voters.

    We now claim that $W = C_1$ satisfies JR.
    Let $N'\subseteq N$ be a $1$-cohesive group with all voters in $N'$ commonly approving a candidate $c\in C\setminus W$.   
    Note that each candidate in $C_1$ is already in $W$, so $c$ must be from $C_2$. 
    Since $c$ is only approved by $\frac nk$ voters, the group $N'$ has to be exactly the support of $c$. 
    This means that $N'$ contains at least one voter $v_i\in N_1$, and there exists $c'\in C_1 = W$ with $c'\in A_{i}$, so $W$ satisfies JR.
    
    Let $W' \neq W$ be another committee satisfying JR.
    Note that every combination of a subset of $N_2$ of size $k^2-1$ together with one voter in $N_1$ forms a $1$-cohesive group. Since $W'$ satisfies JR, at least one of the following two statements holds: (i) every voter in $N_1$ approves at least one candidate from $W'$, or (ii) at most $k^2-2$ voters in $N_2$ approve no candidate from $W'$. Since each candidate in $C_2$ covers only one voter in $N_1$, (i) would imply that $W' = C_1 = W$, a contradiction. Thus, (ii) must hold. Since there exist $k^3-k^2$ voters in $N_2$ and each candidate is approved by at most $k^2-1$ voters in $N_2$, the number of candidates from $C_2$ that we need in $W'$ so that at most $k^2 - 2$ voters in $N_2$ are not covered is at least 
    \begin{align*}
    &\frac{k^3- k^2 -(k^2-2)}{k^2-1} = 
    \frac{k(k^2-1) -(2k^2 - 2) + k}{k^2-1}
    \\
    &= k - 2 + \frac{k}{k^2-1} > k -2. 
    \end{align*} 
    Since this number must be an integer, $W'$ contains at least $k - 1$ candidates from $C_2$; in particular, it has distance at least $k-1$ to $W = C_1$. Finally, one can check that any subset of $C_2$ that covers all voters in $N_2$ indeed satisfies JR, and that such a subset of size $k$ exists.
\end{proof}
Since the instance constructed in the previous proof does not contain $\ell$-cohesive groups for any $\ell > 1$, the statement also applies to EJR.
\begin{corollary}
For any $k \ge 3$, there exists an approval profile $A$ and a committee that is $(k-2)$-isolated in $\EJR(A,k)$.
\end{corollary}

Next, we prove that the bound of $k-2$ in \Cref{thm:isolated_jr} is tight, that is, any JR committee has another JR committee of distance at most $k-1$ to it (provided that the former committee is not the unique JR committee).
Our proof is similar to the proof of \citet[Theorem~3.5]{EFI22b} on the number of committees satisfying JR.

\begin{restatable}{proposition}{prophigh}
\label{prop:isolated}
Suppose that in an instance $(A,k)$, a committee 
is $r$-isolated in $\JR(A,k)$. 
Then, $r\le k-2$.
\end{restatable}
\begin{proof}
Let $W$ be any committee satisfying JR, and assume that there is a JR committee different from $W$ in this instance.
Our goal is to show that there exists a JR committee different from $W$ that has distance at most $k-1$ from $W$.
Construct a subcommittee $W'$ of size $k-1$ by running seqCCAV for $k-1$ steps.
If $W'$ already satisfies JR, we may add another candidate so that the resulting committee $\widehat{W}$ neither coincides with nor is disjoint from $W$, and therefore $1\le d(W, \widehat{W}) \le k-1$.
Assume from now on that $W'$ does not satisfy JR.
This means that each of the $k-1$ candidates in $W'$ covers exactly $\frac{n}{k}$ voters, and the remaining $\frac{n}{k}$ voters are covered by another candidate $x$.
Let $W'' = W'\cup\{x\}$; clearly, $W''$ satisfies JR.
If $1 \le d(W, W'') \le k-1$, we are done.
Otherwise, $W''$ either coincides with or is disjoint from $W$.

\textbf{Case 1:} $W'' = W$.
Since there is a JR committee different from $W = W''$, there exists a candidate $y\not\in W''$ approved by some voter.
Let $z$ be the candidate in $W''$ that covers this voter.
Then, the committee $(W''\setminus\{z\})\cup\{y\}$ satisfies JR and has distance one from $W'' = W$.

\textbf{Case 2:} $W''$ and $W$ are disjoint.
Since $W$ satisfies JR and there exist $1$-cohesive groups in this instance, some candidate $a\in W$ is approved by some voter.
Let $b$ be the candidate in $W''$ that covers this voter.
Then, the committee $(W''\setminus\{b\})\cup\{a\}$ satisfies JR and has distance $k-1$ from $W$.
\end{proof}

On the complexity front, we show that checking whether two committees satisfying JR are connected within the set of JR committees is PSPACE-complete.
To this end, we reduce from the SAT-Reconfiguration problem, which is known to be PSPACE-hard \citep{GKMP09a}. 

\begin{restatable}{theorem}{pspace}\label{thm:PSPACE}
Given an instance $(A,k)$ and two JR committees, deciding whether the two committees are connected in $\JR (A,k)$ is PSPACE-complete.
\end{restatable}
\begin{proof}
    First, we show membership in PSPACE. 
    As PSPACE = NPSPACE \citep{Savi70a}, it is sufficient to solve the problem using a non-deterministic polynomial-space Turing machine. 
    To this end, we simply guess the path between the two committees.
    As the total number of committees is at most ${m \choose k} \le m^k$, if there exists a path between the two given committees that consists only of JR committees, then there exists one with length at most $m^k$.
    Such a path can be verified in polynomial space (i.e., using a polynomial number of bits).

    To establish hardness, we reduce from the PSPACE-hard SAT-Reconfiguration problem \citep{GKMP09a}.
    In this problem, we are given a set $\mathcal{A}$ of $a$ clauses and a set $\mathcal{B}$ of $b$ literals. 
    Each literal can appear positively or negatively. 
    Further, there are two given satisfying assignments $\varphi_1$ and $\varphi_2$, and the goal is to decide whether there exists a path from $\varphi_1$ to $\varphi_2$ consisting only of satisfying assignments, where two assignments are considered neighboring if they differ by only a single bit flip. 
    We can assume without loss of generality that $a + 27$ is divisible by $8$, since we can add copies of existing clauses otherwise. 

    Our JR-Reconfiguration instance consists of the following voters:
    \begin{itemize}
        \item for each clause $A \in \mathcal{A}$, a clause voter $v_A$;
        \item for each literal $x \in \mathcal{B}$, $10$ literal voters $v_x^1, \dots, v_x^{10}$;
        \item $9$ greedy voters $v_{g_1},\dots, v_{g_9}$; 
        \item $18$ special dummy voters $v_{t_1}, \dots v_{t_9}, v_{s_1}, \dots v_{s_9} $,  and $\frac{a + 27}{4}$ non-special dummy voters $v^1_{d_1}, \dots, v^1_{d_{(a + 27)/8}},$ $v^2_{d_1}, \dots, v^2_{d_{(a + 27)/8}}$.
    \end{itemize}
    The total number of voters is $n = a + 10b + 9 + 18 + \frac{a+27}{4}$.
    Further, we have the following candidates:
    \begin{itemize}
        \item for each clause $A \in \mathcal{A}$, a clause candidate $c_A$;
        \item for each literal $x \in \mathcal{B}$, two literal candidates $c_x$ and $c_{\overline{x}}$;
        \item $2$ special dummy candidates $c_1$ and $c_2$, and $\frac{a + 27}{8}$ non-special dummy candidates $c_{d_1}, \dots, c_{d_{(a + 27)/8}}$.
    \end{itemize}
    The approval sets are defined as follows:
    \begin{itemize}
        \item For each clause $A$, the clause voter $v_A$ approves the literal candidates corresponding to the literals in $A$, where the positivity or negativity is taken into account.
        This voter further approves the clause candidate $c_A$.
        \item Each literal voter approves both corresponding literal candidates.
        \item The greedy voters approve all clause candidates.
        \item The special dummy voters $v_{t_1}, \dots v_{t_9}$ and non-special dummy voters $v^1_{d_1}, \dots, v^1_{d_{(a + 27)/8}}$ approve the special dummy candidate $c_1$.
        Similarly, $v_{s_1}, \dots v_{s_9}$ and $v^2_{d_1}, \dots, v^2_{d_{(a + 27)/8}}$ approve $c_2$.
        \item Additionally, for each $j \in [\frac{a+27}{8}]$, the non-special dummy voters $v^1_{d_j}$ and $v^2_{d_j}$ approve the non-special dummy candidate $c_{d_j}$.
    \end{itemize}
    We set $k = b + \frac{a + 27}{8}$. Thus, 
    \begin{align*}
        \frac{n}{k} = \frac{a + 10b + 9 + 18 + \frac{a+27}{4}}{b + \frac{a+27}{8}} = 10.
    \end{align*}
    For the given satisfying assignments $\varphi_1$ and $\varphi_2$, we create corresponding committees consisting of all corresponding literal candidates and all non-special dummy candidates; note that each committee has size $b + \frac{a+27}{8} = k$. 
    Clearly, this reduction can be done in polynomial time.
    To see that both committees satisfy JR, observe that each clause, literal, and non-special dummy voter approves at least one candidate in each committee and therefore cannot be part of any witness of a JR violation. 
    The greedy voters need at least one clause voter to form a $1$-cohesive group, while the special dummy voters need at least one non-special dummy voter. 
    Thus, both committees satisfy JR. 

    If the assignments $\varphi_1$ and $\varphi_2$ are connected via a path of satisfying assignments, then by the same reasoning, the committees corresponding to the assignments on the path also satisfy JR.
    Moreover, each pair of adjacent committees differ only by one literal candidate being swapped for the ``opposite'' literal candidate.

    Conversely, suppose that the two committees are connected via a path of JR committees.
    Observe that to reach another JR committee from one of the two committees, we cannot swap out any non-special dummy candidate, as otherwise one of the two non-special dummy voters approving this candidate together with the $9$ special dummy voters approving the same special dummy candidate as this non-special dummy voter would witness a JR violation. 
    Further, if a literal candidate $c_{x}$ (resp., $c_{\overline{x}}$) is swapped out, it can only be swapped with $c_{\overline{x}}$ (resp., $c_{x}$), as otherwise the $10$ corresponding literal voters would be unsatisfied. 
    Finally, after such a swap, each clause voter must still approve at least one literal candidate, as otherwise the $9$ greedy voters together with this clause candidate would witness a JR violation. 
    Thus, after a swap, the new committee corresponds to a satisfying assignment differing only by the swap of a literal. 
    Therefore, since the two committees are connected, the original assignments are also connected in the SAT-Reconfiguration instance.
\end{proof}

The results so far paint a rather negative picture: proportional committees may not be connected, and determining whether they are connected is computationally intractable. 
Nevertheless, we can obtain positive results if we 
only demand that the intermediate committees satisfy \textit{approximate} proportionality notions.

We first show that any two JR committees are connected in the space of $2$-JR committees. 
To this end, we will need two lemmas. 
The first lemma is an application of the pigeonhole principle---it implies that, given any committee $W$ approved by $|N_W|$ voters and any $s$, we can always remove $s$ candidates so that the resulting subcommittee is still approved by at least $\lvert N_W\rvert - s \cdot \frac{n}{k}$ voters.

\begin{restatable}{lemma}{greedyremove}\label{lem:greedymove}
Let $W$ be any committee. There exists an order $c_1, \dots, c_k$ of the candidates in $W$ such that  $\lvert N_{W\setminus \{c_1,\dots, c_s\}}\rvert \ge \lvert N_W\rvert - s\cdot\frac{n}{k}$ for all $s \in \{0,1,\dots, k\}$.
    \label{claim1}
\end{restatable}
\begin{proof}
We prove the statement by induction on $s$. The claim trivially holds for the induction base $s = 0$. To give a better understanding, we also show the statement for $s=1$. 
Note that there exists a candidate $c_1 \in W$ such that $\lvert N_{W \setminus \{c_1\}} \rvert \geq \lvert N_{W} \rvert -\frac{n}{k}$; otherwise, each of the $k$ candidates in $W$ is uniquely approved by more than $\frac{n}{k}$ voters, contradicting the fact that there are only $ n   \geq \lvert N_{W} \rvert$ voters in total. 
Therefore, the statement also holds for $s=1$.

For the inductive step, let $W_{s-1}$ be the subcommittee $W \setminus \{c_1, \dots, c_{s-1}\}$ for some $s > 1$. 
Then, there exists a candidate $c_s\in W_{s-1}$ such that at least $\frac{k-s}{k-s+1} \cdot \vert N_{W_{s-1}}\rvert$ voters approve at least one candidate in $W_{s} \coloneqq W_{s-1} \setminus \{c_s\}$. 
Indeed, otherwise each of the $k-s+1$ candidates in $W_{s-1}$ is uniquely approved by more than $\frac{1}{k-s+1}\cdot\vert N_{W_{s-1}}\rvert$ voters in $N_{W_{s-1}}$, which implies that $\vert N_{W_{s-1}}\rvert > \vert N_{W_{s-1}}\rvert$, a contradiction. 

By the induction hypothesis, we have $\vert N_{W_{s-1}}\rvert \geq \lvert N_{W} \rvert - (s-1)\cdot\frac{n}{k}$. Thus,

    \begin{align*}
     \lvert N_{W_s} \rvert &\geq \frac{k-s}{k-s + 1}\cdot\left(\lvert N_{W} \rvert - (s-1)\cdot\frac{n}{k}\right) \\ 
     &= 
    \frac{k-s}{k-s + 1}\cdot\left(\lvert N_{W} \rvert - n + n -  (s-1)\cdot\frac{n}{k}\right) \\
    &= \frac{k-s}{k-s+1}\cdot\left(\lvert N_{W} \rvert - n  +  (k- s + 1)\cdot\frac{n}{k}\right) \\
    &= \frac{k-s}{k-s+1}\cdot\left(\lvert N_{W} \rvert - n\right)  + (k-s)\cdot\frac{n}{k} \\ &= \frac{k-s}{k-s+1}\cdot(\lvert N_{W} \rvert - n) + n  -s\cdot\frac{n}{k}
    \\
    &= \frac{k-s}{k-s+1}\cdot\lvert N_{W} \rvert +  \frac{1}{k-s+1}\cdot n  -s\cdot\frac{n}{k} \\
    & \geq \lvert N_{W} \rvert - s\cdot\frac{n}{k},
    \end{align*}
    where the last inequality holds because $n \ge |N_W|$.
    This concludes the proof.
\end{proof}

The second lemma observes that JR ensures 
better representation for groups of voters of size 
at least $2 \cdot \frac{n}{k}$.
\begin{restatable}{lemma}{twojroverlap}\label{lem:twoJRoverlap} For any $W \in \JR(A,k)$ and every $2$-large and $1$-cohesive group $N' \subseteq N$, it holds that $\lvert N'\cap N_W \vert> \frac nk$.
\label{claim2}
\end{restatable}

\begin{proof}
Since $W$ satisfies JR and $N'$ is a $1$-cohesive group, we have $\lvert N'\setminus N_W\rvert < \frac nk$. This implies that $\lvert N'\cap N_W\rvert =\lvert N'\rvert - \lvert N'\setminus N_W\rvert > 2  \cdot\frac nk - \frac nk = \frac nk$, as desired. 
\end{proof}

Using these two lemmas, we will 
show that 
any committee satisfying JR can reach any other committee satisfying JR via a path of $2$-JR committees. 
The high-level idea of the proof is that, by applying \Cref{claim1}, we obtain an 
order in which we 
remove candidates from both committees. With each removed candidate, we decrease the number of voters who approve at least one candidate by at most $\frac{n}{k}$ on average. 
On the other hand, \Cref{claim2} guarantees that any $2$-large, $1$-cohesive 
group has at least $\frac{n}{k}$ represented voters.
Thus, intuitively, the 
added candidates due to $2$-JR violations balance out against the candidates removed from the committees. Since the added candidates represent several voters, this enables us to connect the two resulting committees.
\begin{theorem}
    For any instance $(A,k)$, the set  $\JR(A,k)$ is connected 
    in
    $2$-$\JR(A,k)$. 
    Furthermore, any two JR committees can be connected via a sequence of at most $2k$ committees each satisfying $2$-JR. 
    \label{Theorem:2JR}
    \label{2ConnectivenessTheorem}
\end{theorem}
\begin{proof}
    We say that a subcommittee $W'\subseteq C$ is \textit{$2$-JR greedy} if it satisfies $2$-JR and there is an enumeration $W' = \{d_1,\dots,d_\ell\}$ such that $\lvert N_{d_r}\setminus N_{\{d_1,\dots,d_{r-1}\}} \vert \ge 2 \cdot \frac nk$ for all $d_r\in W'$. That is, $W'$ can be obtained from the empty set by greedily adding new candidates, each covering at least $2\cdot\frac{n}{k}$ previously uncovered voters.
    
    \textbf{Proof outline:} 
    As the first step of the proof, we show that two committees each of which contains a $2$-JR greedy subcommittee are always connected through a path of at most $k$ committees satisfying $2$-JR.
    In the more involved second step, we prove that every committee satisfying JR is connected via at most $\frac k2$ swaps to some superset of a $2$-JR greedy subcommittee. Together, this implies that any two JR committees are connected via at most $2k$ committees satisfying $2$-JR.
    Both steps leverage the fact that $2$-JR greedy subcommittees always exist and contain at most $\frac k2$ candidates.

    \textbf{Step 1:} Let $W$ and $W'$ be two committees, and $W_g\subseteq W$ and $W'_g\subseteq W'$ be $2$-JR greedy subcommittees contained in them.
    By definition, $\lvert W_g\rvert,  \lvert W'_g\rvert \le \frac k2$, and every superset of $W_g$ and $W'_g$ satisfies $2$-JR. 
    Hence, starting with $W$, we can remove candidates from $W\setminus W_g$ and replace them with candidates from $W'_g$ (possibly along with candidates from $W'\setminus W'_g$) without violating $2$-JR. Call the resulting committee $W''$. Since $W'_g\subseteq W''$, we can remove all candidates from $W_g$ and replace them with the remaining candidates from $W'$ without violating $2$-JR. This concludes the proof of Step 1.

    \textbf{Step 2:} Let $W_{\mathrm{JR}}$ be any committee satisfying JR. Our goal is to connect $W_{\mathrm{JR}}$ to some superset of a $2$-JR greedy subcommittee. We first determine a suitable $2$-JR greedy subcommittee $W_g =\{d_1,\dots,d_\ell\}$ through the following process. 
    Enumerate the candidates in $W_{\mathrm{JR}} = \{c_1,\dots,c_k\}$ according to \Cref{claim1}. 
    Initialize $W' = W_{\mathrm{JR}}$ and $r = s = 1$. 
    If $s\le k$, remove $c_s$ from $W'$. 
    If a violation of $2$-JR occurs as a result of the removal, add a witness candidate to $W'$ and label it $d_r$, increment $r$, and repeat until $2$-JR is restored.
    Increment $s$ and repeat starting from the removal.
    Once the process ends, set $\ell = r-1$ and return $W'$.  

    To analyze this process, for each $r\in[\ell]$, denote by $W’_r = (W_{\mathrm{JR}}\setminus \{c_1,\dots,c_s\}) \cup \{d_1,\dots,d_{r-1}\}$ the set for which $d_r$ creates a $2$-JR violation and is added next.
    We claim that (i) for all $W’_r$ created during the process, it holds that $r\le s$; and (ii) when the process terminates, the final subcommittee $W’$ is a $2$-JR greedy subcommittee of the form  $\{d_1,\dots,d_\ell\}$.

    For (i), consider any set $W’_r = \{c_{s+1}, \dots, c_k\} \cup \{d_1,\dots,d_{r-1}\}$, and let $W’= W’_r\cup\{d_r\}$.  
    To show that $r\le s$, we perform a counting argument on $\lvert W'\cap N_{W_{\mathrm{JR}}}\rvert$. 
    In particular, we pretend that we obtain $W'$ by first removing $c_1,\dots,c_s$ from $W_{\mathrm{JR}}$ at once and then iteratively adding $d_1,\dots, d_r$. 
    Formally, let  $W_j = \{c_{s+1}, \dots, c_k\} \cup \{d_1,\dots,d_{j-1}\}$ for each $j\in[r]$. 
    By a telescoping sum argument, we have 
    \begin{align*}
        &\lvert N_{W'}\cap N_{W_{\mathrm{JR}}}\rvert \\
        &= (\lvert N_{W'}\cap N_{W_{\mathrm{JR}}}\rvert - \lvert N_{W_r}\cap N_{W_{\mathrm{JR}}}\rvert) \\
        &\quad + (\lvert N_{W_r}\cap N_{W_{\mathrm{JR}}}\rvert - \lvert N_{W_{r-1}}\cap N_{W_{\mathrm{JR}}}\rvert) \\
        &\quad + \cdots  \\
        &\quad + (\lvert N_{W_2}\cap N_{W_{\mathrm{JR}}}\rvert - \lvert N_{\{c_{s+1},\dots, c_k\}}\cap N_{W_{\mathrm{JR}}}\rvert) \\
        &\quad + \lvert N_{\{c_{s+1},\dots, c_k\}}\cap N_{W_{\mathrm{JR}}}\rvert,
    \end{align*}
    where $\{c_{s+1},\dots,c_k\} = W_1$.

    We claim that $\lvert (N_{d_j}\setminus N_{W_j}) \cap  N_{W_{\mathrm{JR}}}\rvert >\frac{n}{k}$ for each $j\in [r]$.    
    By the $2$-JR violation of $W'_{j}$, the set $N_{d_j}$ contains a $1$-cohesive subset $N^*$ of at least $2\cdot\frac nk$ voters that are not in $N_{W'_j}$. 
    Since ${W_j}\subseteq {W'_j}$, we have $N_{W_j}\subseteq N_{W'_j}$, and therefore $N^*\cap N_{W_j} = \emptyset$ too.
    By \Cref{lem:twoJRoverlap}, it holds that $\lvert N^*\cap N_{W_{\mathrm{JR}}}\rvert >\frac nk$, and so $\lvert (N_{d_j}\setminus N_{W_j}) \cap  N_{W_{\mathrm{JR}}}\rvert >\frac{n}{k}$, as claimed. 
    Further, note that $\lvert N_{ \{c_{s+1}, \dots, c_k\}}\cap N_{W_{\mathrm{JR}}}\rvert \ge\lvert N_{W_{\mathrm{JR}}} \rvert - s\cdot\frac{n}{k}$ due to Lemma \ref{claim1}.     
    Hence, by the telescoping relation above, we obtain that $\lvert N_{W'}\cap N_{W_{\mathrm{JR}}}\rvert > \lvert N_{W_{\mathrm{JR}}} \rvert - s\cdot\frac{n}{k} + r\cdot \frac{n}{k}$. 
    As $\lvert N_{W'}\cap N_{W_{\mathrm{JR}}}\rvert \le \lvert N_{W_{\mathrm{JR}}}\rvert$, it follows that $r<s$, concluding the proof of (i).

    To establish (ii), note that the final subcommittee is of the form $W' = (W_{\mathrm{JR}}\setminus \{c_1,\dots,c_k\}) \cup \{d_1,\dots,d_\ell\} = \{d_1,\dots,d_\ell\}$. Since each $d_r$ was added as a witness for a $2$-JR violation of some superset of $W_r$, we have that the candidate $d_r$ satisfies $\lvert N_{d_r}\setminus N_{\{d_1,\dots,d_{r-1}\}}\rvert\ge \lvert N_{d_r}\setminus N_{W_r}\rvert \ge 2\cdot\frac nk$.  
    Hence, the process ends with some $2$-JR greedy subcommittee, concluding the proof of (ii).    
    
    Finally, we specify the desired sequence by letting $W''_r = (W_{\mathrm{JR}}  \setminus \{c_1,\dots,c_r\}) \cup \{d_1,\dots,d_r\}$ for each $r\in[\ell]$. 
    By (ii), $W''_\ell$ is a committee containing a $2$-JR greedy subcommittee. Further, by (i), each committee $W''_r$ is a superset of some $2$-JR subcommittee created during the process. 
    Hence, $W''_r$ satisfies $2$-JR. 
    Since $d(W_{\mathrm{JR}},W''_1) = d(W''_1,W''_2) = \dots = d(W''_{\ell-1},W''_\ell) = 1$ and $\ell\le \frac k2$ by (ii), this completes the proof of Step~2, and therefore the entire proof. \qedhere
\end{proof}

Theorems \ref{thm:isolated_jr}  and \ref{2ConnectivenessTheorem} show that two committees satisfying JR may not be connected in the set of JR committees, but are always connected in the set of $2$-JR committees.  
This raises the question of whether the factor $2$ is tight.
The following proposition answers this question in the affirmative.
\begin{restatable}{proposition}{alphabettertwoimp}\label{lem:2isTight}
For any $\alpha < 2$, there exists an instance $(A,k)$ such that $\JR(A, k)$ is not connected 
in
$\alpha$-$\JR(A,k)$.
\end{restatable}
\begin{proof}
    This proof follows a similar idea as that of \Cref{thm:isolated_jr}.  
    Our goal is to show that for each integer $r\ge 3$ and $\alpha_r = \frac{2r}{r+2}$, there exists a profile $A$ and a target size $k$ such that $\JR(A,k)$ is not connected in $\alpha_r\text{-}\JR(A,k)$. 
    Since $\alpha_r \to 2$ for $r\to\infty$, this is sufficient to prove the proposition. 
    Fix $r\ge 3$, and let $\alpha = \frac{2r}{r+2}$.
	
	\textbf{Construction of a class of profiles:}
     Let the number of voters be $n = r(r+1)^2+ r(r+1)$, and the committee size be $k = r(r+1)$. Hence, $\frac{n}{k} = r+2$. 
    Let $N = N_1 \cup N_2$ be the set of voters partitioned into two parts, with $r (r+1)^2$ voters in $N_1 = \{v_1, \dots, v_{n-r (r+1)}\}$ and $r (r+1)$ voters in $N_2 = \{v_{n-r (r+1)+1}, \dots, v_n\}$. 
    Furthermore, let $C = C_1 \cup C_2$ be the set of candidates, which we partition into a weak set ($C_1$) and a strong set $(C_2)$.
    We choose $C_1 = \{c_1,\dots, c_{k} \}$ such that $c_i$ is approved by $r+1$ consecutive voters $v_{ 1+ (i-1)(r+1) }$, $
    v_{2 + (i-1) (r+1) }$, $ \dots, v_{(r+1) + (i-1) (r+1) }$. 
    Note that the candidates in $C_1$ together cover the voters in $N_1$.
    
    Additionally, for every subset $N'_1 \subseteq N_1$ with $\lvert N'_1 \rvert = r$ and every subset $N_2'\subseteq N_2$ with $\lvert N_2'\rvert = r$, we create a candidate $d_{N_1',N_2'}$  approved by every voter in $N_1' \cup N_2'$. Formally, we set $C_2 =\{d_{N_1',N_2'} \mid N_1'\in \binom{N_1}{r},\, N_2'\in \binom{N_2}{r}\}$. 

    \textbf{Choosing two JR committees:}
    We now claim that $W = C_1$ satisfies JR.
    Let $N'\subseteq N$ be a $1$-cohesive group with all voters in $N'$ commonly approving some candidate $c\in C$. 
    Each candidate in $C_1$ has a support of fewer than $ \frac nk$ voters while each candidate in $C_2$ has a support of exactly $ 2r $ voters, so $c\in C_2$. 
    Observe that the support of each candidate in $C_2$ consists of precisely $r$ voters from $N_1$ and $r$ voters from $N_2$, where $r < r+2 = \frac{n}{k}$. 
    Since $N'\subseteq N_c$ is of size at least $\frac nk$, we have $N'\cap N_1\neq \emptyset$. 
    Since $N_W = N_1$, it holds that $N'\cap N_W\neq \emptyset$, and so $W$ satisfies JR.
    
    Next, we 
	construct
	a committee $W'$ containing at least $r+1$ candidates from $C_2$ that also satisfies JR. 
    Since $|N_2| = r (r+1)$, we can choose a partition $N_2^1, \dots, N_2^{r+1}$ of $N_2$ such that $\lvert N_2^i \rvert = r$ for each $i \in [r+1]$. 
    Additionally, choose any $N_1' \subseteq N_1$ with $\lvert N_1' \rvert = r$. 
    We claim that the subcommittee $D = \{d_{N_1',N_2^1},d_{N_1',N_2^2}, \dots, d_{N_1',N_2^{r+1}} \}$ of size $r+1$ satisfies JR. 
    As already discussed for $W$, each $1$-cohesive group contains at least $r+2$ voters, of which at least two are from $N_2$. 
    Since every voter in $N_2$ approves at least one candidate in $D$, this implies that $D$ satisfies JR. 
    We now extend $D$ with any $k-(r+1)$ candidates to a committee $W' \supseteq D$. 
    This committee $W'$ still satisfies JR and contains at least $r+1$ candidates from $C_2$.

    \textbf{A large set of committees violating $\alpha$-JR:} 
	Next, we claim that any $\alpha$-JR committee $W''$ that contains exactly $r$ candidates from $C_2$ violates $\alpha$-JR. 
	
	Let such a committee $W''$ be given.
	Recall that $\alpha = \frac {2r}{r+2}$, so we need to find a $1$-cohesive group $N'$ of size at least $\frac {2r}{r+2}\cdot(r+2) = 2r$ with $N'\cap N_{W''}=\emptyset$.
	Note that $\lvert N_{W''} \cap N_2 \rvert \leq r^2$.
    Since $|N_2| = r(r+1)$, there exists a subset of $r$ voters $N_2' \subseteq N_2$ that do not approve any candidate in $W''$.
    Moreover, it holds that $\lvert N_{W''} \cap N_1 \rvert \leq r^2 + (k-r)(r+1) = r^2 + r^2(r+1) = r^2(r+2) = r(r+1)^2-r$. 
    Therefore, there exists a subset of $r$ voters $N_1'\subseteq N_1$ that do not approve any candidate in $W''$. 
    Letting $N' = N_1' \cup N_2'$, we have $\lvert N'\rvert = 2r$. 
    By construction, all voters in $N'$ approve the candidate $d_{N_1', N_2'}$, so $W''$ indeed violates $\alpha$-JR.

	Since $W'$ contains at least $r+1$ candidates from $C_2$ while $W$ contains $0$ candidates from $C_2$, any path connecting $W$ and $W'$ must contain some committee $W''$ containing $r$ candidates from $C_2$.
    However, this committee necessarily violates $\alpha$-JR. 	
    It follows that $W$ and $W'$ cannot be connected in $\alpha$-$\JR(A,k)$ and hence concludes the proof.\qedhere
	
\end{proof}

Furthermore, we can employ a similar technique as in Theorem \ref{2ConnectivenessTheorem} to show that any two EJR committees are connected via $4$-EJR committees. To achieve this, we switch from using the number of ``represented'' voters as our metric for deleting candidates to a variant of the PAV-score. The proof idea remains the same: we greedily delete the candidates that decrease this variant of the PAV score the least and replace them by candidates witnessing a $4$-EJR violation.
In the end, this will allow us to reach a common $4$-EJR subcommittee from any EJR committee as a starting point.
\begin{restatable}{theorem}{EJRConnectedInFourEJR}\label{thm:EJRinFourEJR}
    For any instance $(A,k)$, the set $\EJR(A,k)$ is connected in $4$-$\EJR(A,k)$.
\end{restatable}
\begin{proof}
    Let any committee $W$ satisfying EJR be given.

    \textbf{Proof idea:} \citet[Lemma~5.5]{DBW+24a} have shown that $4$-EJR can always be satisfied by a subcommittee $W_{4\text{-}\mathrm{EJR}}$ of size $\lfloor \frac k4\rfloor$.\footnote{Our notion $\alpha$-EJR is called $\frac{1}{\alpha}$-EJR in their paper.}
     Our goal is to show that by removing candidates $c_i$ from $W$ and instead adding witnesses of $4$-EJR violations $d_j$, 
     we remove $\frac{k}{4}$ more candidates than we need to add back while still maintaining $4$-EJR.
     Therefore, for the remaining $\frac{k}{4}$ slots, we can add the candidates from $W_{4\text{-}\mathrm{EJR}}$.
     For any other $\overline{W}$ satisfying EJR, we can apply the same procedure. 
     We can then make $3 \cdot\frac k4$ swaps to connect the two committees containing $W_{4\text{-}\mathrm{EJR}}$, which establishes the connectedness of $W$ and $\overline{W}$.

    \textbf{Notation:} To guarantee that our intermediate committees satisfy $4$-EJR, note that any subcommittee $W'$ with $\lvert A_v\cap W\rvert\le \lvert A_v\cap W'\rvert$ for all $v\in N$ satisfies EJR. We define the \textit{modified PAV score} as follows: 
    \[\PAV_W(W') \coloneqq \sum_{v \in N} H_{\min(\lvert A_v \cap W \rvert, \lvert A_v \cap W' \rvert )}.\]    
    Clearly, it holds for all candidate sets $W'\subseteq C$ that $\PAV_W(W')\le \PAV_W(W)$. Further, $ \PAV_W(W') = \PAV_W(W)$ implies that $W'$ satisfies EJR.
    For a subset $W'\subseteq C$ and candidate $c\in W'$, we define the \textit{marginal contribution} of $c$ to $W'$ as $\Delta_W(W',c) \coloneqq \PAV_W(W') - \PAV_W(W'\setminus\{c\}) = \sum_{v \in N_c: \lvert A_v\cap W'\rvert\le \lvert A_v\cap W\rvert} \frac{1}{\lvert A_v \cap W' \rvert}$.
    Note that if $W'\subseteq W''$, then $\Delta_W(W',c) \ge \Delta_W(W'',c)$.
    
    In the following, we first determine which candidates to remove from $W$, then determine which candidates to add. For both steps, the subcommittees are only considered for convenience of calculation, not because they truly appear in the swap sequence. Only after that will we specify the order in which we truly add and delete, and prove that each resulting committee satisfies $4$-EJR. 
    
    \textbf{Removing candidates from $W$:}
    We first determine the order of deletions of $r \coloneqq \lfloor \frac 23  k \rfloor$ candidates in $W$. 
    Note that 
    \begin{align*}
    \sum_{c\in W}\Delta_W(W,c) 
    &= \sum_{c\in W} \sum_{v \in N_c} \frac{1}{\lvert A_v \cap W \rvert} \\
    &= \sum_{v \in N_W} \sum_{c \in W\cap A_v} \frac{1}{\lvert A_v \cap W \rvert} \\
    &= \sum_{v \in N_W} 1 = |N_W| \le n. 
    \end{align*}
    Hence, there exists $c_1\in W$ with $\Delta_W(W,c_1)\le \frac nk$.
    Inductively, for any $s < k$, since $(W\setminus\{c_1,\dots, c_s\}) \subseteq W$, we obtain that there is a candidate $c_{s+1}$ with $\Delta_W(W\setminus\{c_1,\dots, c_s\}, c_{s+1})\le \frac{n}{k-s}$.
    Using telescoping, we obtain for each $i\in [r]$ that  
    \begin{align}
        &\PAV_W(W\setminus\{c_1,\dots,c_i\}) \nonumber \\
        &= \PAV_W(W) - \sum_{j\in [i]} \Delta_W(W\setminus\{c_1,\dots,c_{j-1}\}, c_j) \nonumber  \\
        &\ge \PAV_W(W) - \sum_{j\in [i]} \frac{n}{k-j+1} \nonumber \\ 
        &= \PAV_W(W) - n(H_k - H_{k - i}). \label{eq:PAV-difference}
    \end{align}
    
    \textbf{Adding candidates to $W$:} Next, we determine the candidates with which we replace the candidates $c_i$.
    For this, initialize $W' = W$ and $i,j = 1$. 
    Remove $c_i$ from $W'$. 
    As long as there exist $4$-EJR violations, we add a witness candidate $d_j$ to the subcommittee and increment $j$. 
    When there are no more $4$-EJR violations, increment $i$ and repeat the process starting from the removal.
    The process ends when we have removed $r$ candidates from $W$ and restored $4$-EJR. 
	
	\textbf{Claim 1:} For the final set $\widehat W=(W\setminus \{c_1,\dots,c_r\})\cup \{d_1,\dots,d_s\}$, it holds that $r-s \ge \frac k4 -1$. 
	
    To avoid lengthy notation, we set $W_0 = W\setminus \{c_1,\dots,c_r\}$ and $W_j = W_0 \cup \{d_1,\dots,d_j\}$ for $j>0$. 
    Then, $\Delta_W(W_j, d_j) = \PAV_W(W_j) - \PAV_W(W_{j-1})$ is the marginal contribution of $d_j$. 
    As the candidate $d_j$ is part of a $4$-EJR violation for some intermediate committee $W'$, there exists a set $N'$ of size at least $4\ell\cdot\frac{n}{k}$ for some $\ell$ witnessing such a violation for $W'$ together with a set of $\ell$ candidates containing $d_j$. 
    By definition, all voters in $N'$ approve strictly less than $\ell$ candidates from $W'$, while strictly more than $3\ell \cdot\frac nk$ voters in $N'$ approve at least $\ell$ candidates from $W$ since $W$ satisfies EJR.
    Thus, each of these latter (more than $3\ell \cdot\frac nk$) voters approves more candidates in $W$ than in $W'$, and therefore contributes at least $\frac{1}{\ell}$ to the modified PAV score.
    Thus, we obtain $\Delta_W(W'\cup\{d_j\}, d_j) > 3 \ell\cdot\frac nk\cdot \frac 1\ell = 3\cdot\frac nk$.
    Since $W_j\subseteq W'\cup\{d_j\}$, it holds that $\Delta_W(W_j,d_j) \ge \Delta_W(W'\cup\{d_j\}, d_j) > 3\cdot\frac{n}{k}$.
    Using this, we can now bound
    \begin{align*}
        \PAV_W(\widehat W) &= \PAV_W(W_s)\\
        &= \PAV_W(W_0) \\ &\quad+ [\PAV_W(W_1) - \PAV_W(W_0)] \\
        &\quad+ [\PAV_W(W_2) - \PAV_W(W_1)]+\cdots\\
        &\quad+ [\PAV_W(W_s) - \PAV_W(W_{s-1})] \\
        &= \PAV_W(W_0) + \sum_{j\in [s]} \Delta_W(W_{j},d_j)\\ 
        &\ge \PAV_W(W) - n(H_k - H_{k - r}) + 3s\cdot\frac nk. 
    \end{align*}
    
	Since $r \le \frac 23  k$, we have
    \begin{align*}
    H_k - H_{k - r} 
    &= \frac{1}{k-r+1} + \frac{1}{k-r+2} + \cdots + \frac{1}{k} \\
    &\le \int_{k-r+1}^{k+1} \frac{1}{x-1} \diff x \\
    &= \int_{k-r}^k \frac{1}{x} \diff x \\
    &= \ln(k) - \ln(k-r) \\
    &\le \ln(k) - \ln\left(k-\frac 23  k\right) = \ln(3).
    \end{align*}
It follows that 
	    \begin{align*}
	         \PAV_W(\widehat W) & \ge \PAV_W(W) - n(H_k - H_{k - r}) + 3s\cdot\frac nk\\
	         & \ge \PAV_W(W) + 3s\cdot\frac{n}{k} - \ln(3)\cdot n
	    \end{align*}
    Thus, $\PAV_W(W)\ge \PAV_W(\widehat W) \ge \PAV_W(W) -\ln(3)\cdot n + 3s\cdot\frac nk$, and therefore $ \frac{\ln(3)}{3} \ge  \frac sk$. Rearranging gives us $s \leq \frac{\ln(3)}{3}\cdot k$. 
    
    Hence, the set $\widehat W$ has size at most $k - \lfloor \frac{2}{3} k  \rfloor + \lfloor \frac{\ln(3)}{3}\cdot k \rfloor 
    \leq k - \frac{2}{3} k + 1 + \frac{\ln(3)}{3}\cdot k 
    = \frac{\ln(3) + 1}{3}\cdot k +1 \leq \frac{3}{4} k + 1$.
    This concludes the proof of Claim 1.
    
    \textbf{Claim 2:} It holds that $i-j\ge 0$ at any point during the process of adding candidates to $W$.
    That is, the set $W'$ never exceeds the committee size $k$.
    
	Let any $i\in [r]$ be given. By \eqref{eq:PAV-difference}, we have 
    \begin{align*}
    \PAV_W(W\setminus\{c_1,\dots,c_i\}) &\ge \PAV_W(W) - i \cdot\frac{n}{k-i +1} \\
    &\ge \PAV_W(W) - i \cdot\frac{n}{\frac 13 k +1} \\
    &> \PAV_W(W) - 3i\cdot\frac{n}{k}.
    \end{align*}
	Further, as argued in the proof of Claim~1, each time we add a candidate, the $\PAV_W$ score increases by more than $3\cdot \frac{n}{k}$. 
    This implies that
    \begin{align*}
    \PAV_W(W) &\geq \PAV_W((W \setminus\{c_1,\dots,c_i\}) \cup \{d_1, \dots, d_j\}) \\
    &\geq \PAV_W(W) - 3i\cdot\frac{n}{k} + 3j\cdot\frac{n}{k},
    \end{align*}
    where the first inequality follows from the fact that $\PAV_W(W) \ge \PAV_W(W')$ for any $W'$, and for the second inequality we use the fact that $\Delta_W(W',c) \ge \Delta_W(W'',c)$ whenever $W'\subseteq W''$.
    Therefore, it holds that $- 3i\cdot\frac{n}{k} + 3j \cdot\frac nk \le 0$, which means that $i\ge j$.
    This concludes the proof of Claim~2.

    \textbf{Construction of the transition committees:}
    Finally, we specify the intermediate committees that connect the committee $W$ to the subcommittee $W_{4\text{-}\mathrm{EJR}}$.
    We enumerate the candidates of our target subcommittee as $W_{4\text{-}\mathrm{EJR}} = \{d_{s+1},\dots,d_{s+\lfloor \frac k4\rfloor}\}$.
    We begin with $W^0\coloneqq W$ and set $W^i= (W\setminus \{c_1,\dots,c_i\})\cup \{d_1,\dots,d_i\}$ for all $i\le s+\lfloor \frac k4\rfloor-1$. 
    As we have shown in Claim 1, $r\ge s+\lfloor \frac k4\rfloor-1$, so the operations are well-defined since all required candidates $c_i$ exist.
    By definition, it holds that $d(W^i,W^{i+1}) = 1$ for all $i$. 
    Further, to prove that each $W^i$ satisfies $4$-EJR, note that some committee $W'= (W\setminus \{c_1,\dots,c_i\})\cup \{d_1,\dots,d_j\}$ satisfies $4$-EJR. 
    By Claim 2, we have $i\ge j$; this means that $W'\subseteq W^i$, and so $W^i$ satisfies $4$-EJR. 
    As the final step, note that for $i^* = s+\lfloor \frac k4\rfloor-1$, we still have $W_{4\text{-}\mathrm{EJR}}\not\subseteq W^{i^*}$, as $d_{i^*+1}\notin W^{i^*}$. However, we can define $W^{i^*+1}$ by replacing any candidate from $W^{i^*}\setminus W_{4\text{-}\mathrm{EJR}}$ with $d_{i^*+1}$. This concludes the proof that $W$ is connected to the subcommittee $W_{4\text{-}\mathrm{EJR}}$ in $4$-$\EJR(A,k)$. Since $W$ was an arbitrary committee satisfying EJR, it holds that $\EJR(A,k)$ is connected in $4$-$\EJR(A,k)$, proving the theorem. \qedhere 
	
\end{proof}

\section{Specific Voting Rules}
\label{sec:specific-rules}

In \Cref{thm:isolated_jr}, we showed that proportional committees can be ``far away'' from other proportional committees. 
However, the committee in that example appears to be rather suboptimal, and would not be selected by common voting rules such as PAV or MES.
In this section, we examine how well the choice sets of well-known proportional voting rules---that is, the sets of committees returned by these rules---are connected, either to other proportional committees or within themselves.

Our first observation is that connectedness within the choice set itself is not a reasonable demand for approval-based multiwinner voting rules, as the choice sets may be too sparse to allow for connections. 
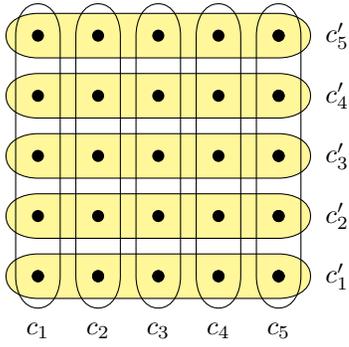
\begin{figure}
\centering
\begin{tikzpicture}[scale=0.8]
    \tikzmath{
    \height = 0.37;
    \width = 0.37;
    \cornerpoint = 0.53;
    \numberVertices = 5;
    }
    \foreach \v in {1,...,\numberVertices} {
        \foreach \w in {1,...,\numberVertices} {
            \node (v\v-\w) at (\w,\v) {};
        }
    }

    \begin{scope}[fill opacity=0.8]
    \foreach \v in {1,...,\numberVertices} {
    \filldraw[fill=yellow!50] ($(v\v-1)+(-\cornerpoint,0)$) 
        to[out=90,in=180] ($(v\v-1) + (0,\height)$)
        to[out=0,in=180] ($(v\v-\numberVertices) + (0,\height)$)
        to[out=0,in=90] ($(v\v-\numberVertices) + (\cornerpoint,0)$)
        to[out=270,in=0] ($(v\v-\numberVertices) + (0,-\height)$)
        to[out=180,in=0] ($(v\v-1) + (0,-\height)$)
        to[out=180,in=270] ($(v\v-1) + (-\cornerpoint,0)$);
        }
    \foreach \v in {1,...,\numberVertices} {
        \draw ($(v1-\v)+(0,-\cornerpoint)$) 
        to[out=0,in=270] ($(v1-\v) + (\width,0)$)
        to[out=90,in=270] ($(v\numberVertices-\v) + (\width,0)$)
        to[out=90,in=0] ($(v\numberVertices-\v) + (0,\cornerpoint)$)
        to[out=180,in=90] ($(v\numberVertices-\v) + (-\width,0)$)
        to[out=270,in=90] ($(v1-\v) + (-\width,0)$)
        to[out=270,in=180] ($(v1-\v) + (0,-\cornerpoint)$);
        }
    \end{scope}
    \foreach \v in {1,...,\numberVertices} {
        \foreach \w in {1,...,\numberVertices} {
             \fill (v\v-\w) circle (0.1);
        }
    }

    \foreach \v in {1,...,\numberVertices} {
        \fill (v1-\v) circle (0.1) node  [below = 5 mm] {$c_{\v}$};
    }
    \foreach \w in {1,...,\numberVertices} {
        \fill (v\w-\numberVertices) circle (0.1) node [right = 5 mm ] {$c'_{\w}$};
    }

\end{tikzpicture}
\caption{A visualization of Example \ref{exp:choice_set_disc}. Voters are depicted as nodes, and row and column candidates as hyperedges around the nodes. The row committee $\{c'_i\mid 1\le i\le 5\}$ is shown in yellow.}
\label{fig:disc}
\end{figure}
\begin{example}
\label{exp:choice_set_disc}
\normalfont
Consider an instance with $r^2$ voters corresponding to points in an $r \times r$ grid. There are $2r$ candidates, one for each row and one for each column.
Each voter approves the two candidates representing her column and row. 
For the committee size $k = r$, most well-known proportional voting rules (such as PAV, MES, and sequential-Phragm\'en) select only the committee consisting of all row candidates and the committee consisting of all column candidates. However, these two committees are not connected in the respective choice sets.
See \Cref{fig:disc} for a visualization.
\end{example}

In light of this observation, we instead ask whether a choice set is connected in $\JR(A,k)$ or $\EJR(A,k)$, and whether isolation can be used as an indicator of ``bad'' EJR committees.

As our first positive result, we show that a large and well-behaved set of subcommittees satisfying JR is connected. To specify this set, we recap the following 
variant of priceability \citep{PeSk20a} called \emph{affordability}, introduced by \citet{BrPe24a}. 
We state a more succinct version of it here.

\begin{definition}[\citealp{BrPe24a}]
\label{def:affordable}
    A subcommittee $W$ is \emph{affordable} if there exists a family of payment functions $(p_i \colon A_i \to [0,1])_{v_i \in N}$ such that 
    \begin{itemize}
        \item $\sum_{c \in A_i} p_i(c) \le \frac kn$ for each $v_i \in N$, and
        \item $\sum_{v_i \in N _c} p_i(c) = 1 $ for each $c \in W$.
    \end{itemize}
\end{definition}

Furthermore, we define $p$ to be a \emph{payment system} for an affordable committee $W$ if it is a family of payment functions that uphold the conditions in \Cref{def:affordable}.

We can now express our goal more precisely: we want to show that the set of all affordable subcommittees satisfying JR is connected via JR committees. 
As it turns out, this will have a favorable impact on the connectedness of several voting rules.
Before proceeding, we need two auxiliary lemmas.
Recall that any subset of an affordable subcommittee is also an affordable subcommittee, and the empty subcommittee is affordable \citep{BrPe24a}.
\begin{lemma}[{\citealp[Observation~1]{BrPe24a}}]
    Let $W_\emph{aff}$ be an affordable subcommittee. Then, for all subcommittees $X \subseteq W_\emph{aff}$, it holds that $\lvert N_X\rvert\geq |X| \cdot\frac{n}{k}$.
    \label{Lemma:SubcommitteesAlgorithm}
\end{lemma}
Secondly, we show that we can connect any subcommittee $W$ that satisfies JR and is approved by at least $\frac{n}{k}\cdot(\lvert W\rvert -1)$ voters to \emph{some} affordable subcommittee satisfying JR and containing an affordable subcommittee of $W$.

\begin{restatable}{lemma}{connectlemma}
    Let $W$ be a subcommittee that satisfies JR such that $\vert N_W \vert \geq \frac{n}{k}\cdot( \vert W \vert - 1) $,  and let $W'_\emph{aff}$ be an affordable subcommittee. Further, let  $X \subseteq W\cap W'_\emph{aff}$. Then, there exists an affordable subcommittee $W^X_\emph{aff}$ satisfying JR such that $X \subseteq W^X_\emph{aff}$, and $W$ and $W^X_\emph{aff}$ are connected in $\JR(A,k)$. \label{Lemma:SubCommitteeExtension}
\end{restatable}
\begin{proof}
We describe a procedure for reaching such a subcommittee $W^X_\text{aff}$. 
Note that $X$ is affordable, since $W'_\text{aff}$ is affordable and $X\subseteq W'_\text{aff}$.
Consider an inclusion-maximal affordable set of candidates $X^*$ such that $X\subseteq X^* \subseteq W$.
From Lemma \ref{Lemma:SubcommitteesAlgorithm}, we know that $\vert N_{X^*} \vert \geq \lvert X^*\rvert \cdot\frac nk$. 

Enumerate the candidates in $W \setminus X^*$ as  $c_1, \dots, c_r$.
Set $W^X_0 = W$, which satisfies JR by assumption, and let $i = j = 1$.
Initialize $ W^X_i$ to $W^X_{i-1} \setminus \{c_i\}$. 
If $W^X_i$ does not satisfy JR, add a witness candidate $d_j$ to $W^X_i$ and increment $j$; repeat this until there are no JR violations.
Then, increment $i$ and repeat the process. 
The process ends when we have removed $c_r$ and restored JR.

We show that for all $i \in [r]$, it holds that $W^X_i$ never exceeds the committee size $k$. 
Specifically, we claim that if $W^X_{i} = (W \setminus \{c_1, \dots,c_i\}) \cup \{d_1, \dots,d_s\}$ after JR is restored, then $s\le k-\vert W\vert +i$; this implies that $|W^X_i| = |W|-i+s\le k$. 
To prove this claim, first note that for any removed candidate $c_j$, we must have $\lvert N_{c_j} \setminus N_{X^*}  \rvert < \frac{n}{k}$; otherwise, the voters who approve $c_j$ but no candidate in $X^*$ can spend at most $\frac{k}{n}$ each on $c_j$, and $X^*\cup\{c_j\}$ would be affordable, contradicting the inclusion-maximality of $X^*$. 
Since for any candidate $c \in C$ and any sets $Y' \subseteq Y$ it holds that $\lvert N_{c} \setminus N_Y \rvert \leq \lvert N_{c} \setminus N_{Y'} \rvert$, this implies that 
\begin{align*}
&\lvert N_{\{c_1,\dots, c_i\}}\setminus N_{W\setminus\{c_1,\dots, c_i\}}\rvert \\
&= \lvert N_{c_1}\setminus N_{W\setminus\{c_1\}}\rvert \\
&\quad + \lvert N_{c_2}\setminus N_{W \setminus\{c_1, c_2\}}\rvert + \dots + \lvert N_{c_i}\setminus N_{W  \setminus\{c_1,\dots, c_i\}}\rvert \\
&\leq \lvert N_{c_1}\setminus N_{X^*}\rvert + \lvert N_{c_2}\setminus N_{X^*}\rvert + \dots + \lvert N_{c_i}\setminus N_{X^*}\rvert \\
&< i \cdot \frac{n}{k}.
\end{align*}
In particular, for the first transition, observe that if a voter belongs to $N_{\{c_1,\dots, c_i\}}\setminus N_{W\setminus\{c_1,\dots, c_i\}}$ and $c_j$ is the highest-index candidate among $c_1,\dots,c_i$ that the voter approves, then the only index $\ell\in[i]$ such that the voter belongs to $N_{c_\ell}\setminus N_{W\setminus\{c_1,\dots,c_\ell\}}$ is $\ell = j$.
It follows that $|N_W| - |N_{W\setminus\{c_1,\dots,c_i\}}| < i\cdot\frac{n}{k}$.

Next, every time we add a candidate due to a JR violation, the number of voters that approve at least one candidate in the committee increases by at least $\frac{n}{k}$.

Since for any candidate $c \in C$ and any sets $Y' \subseteq Y$ it holds that $\lvert N_{c} \setminus N_Y \rvert \leq \lvert N_{c} \setminus N_{Y'} \rvert$, we have

\begin{align*}
&\lvert N_{\{d_1,\dots, d_s\}}\setminus N_{W\setminus\{c_1,\dots, c_i\}}\rvert \\
&= \lvert N_{d_1}\setminus N_{W\setminus\{c_1,\dots, c_i\}}\rvert + \lvert N_{d_2}\setminus N_{(W \cup \{d_1\}) \setminus\{c_1,\dots, c_i\}}\rvert \\
&\quad+ \dots + \lvert N_{d_s}\setminus N_{(W \cup \{d_1,d_2,\dots, d_{s-1}\}) \setminus\{c_1,\dots, c_i\}}\rvert \\
&\geq s \cdot \frac{n}{k}.
\end{align*}

Using these facts and viewing $\lvert N_{W^X_i}\rvert = \lvert N_{W\setminus\{c_1,\dots, c_i\}}\rvert + \lvert N_{\{d_1,\dots, d_s\}}\setminus N_{W\setminus\{c_1,\dots, c_i\}}\rvert$, we obtain $n \ge \vert N_{W^X_i} \vert > (\vert N_{W}\vert - i \cdot\frac{n}{k}) + s\cdot\frac{n}{k}$.
Hence, $k > \vert N_{W}\vert \cdot\frac kn - i + s $. 
Since $\vert N_W \vert \geq \frac{n}{k}\cdot( \vert W \vert - 1) $ by assumption, we obtain $k >  ( \vert W \vert - 1) - i + s $. 
Therefore, $s < k-|W|+i+1$, which means that $s\le k-|W|+i$, as claimed.

Let $W^X_\text{aff} = W^X_r$.
In order to show that the subcommittees $W$ and $W^X_\text{aff}$ are connected, it suffices to show that all committees containing them are connected.
To this end, let $W'$ be any committee such that $W \subseteq W'$, and let $\widehat{W}$ be any committee such that $W^X_\text{aff} \subseteq  \widehat{W}$. 
We now specify a sequence of committees $W_0,W_1,\dots, W_x \in \JR(A,k)$ such that $W_0 = W'$, $W_x = \widehat{W}$, and $d(W_{i-1}, W_{i}) = 1$ for all $i\in[x]$.

First, we define the sequence of candidates that we remove.
Let $W' \setminus W = \{y_1, \dots, y_t\}$.
Then, we define the sequence of candidates $a_1, \dots, a_{t+r}$ by letting $a_i = y_i$ for $i \in [t]$ and $a_i = c_{i-t}$ for $i \in \{t+1, \dots, t+r\}$.

Next, we define the sequence of candidates that we add.
Let $\widehat{W} \setminus W^X_\text{aff} = \{z_1, \dots, z_u\}$.
Since $|W'| = |\widehat{W}| = k$, and $W' = X^* \cup \{c_1, \dots, c_r\} \cup \{y_1, \dots, y_t\}$ while $\widehat{W} = X^* \cup \{d_1, \dots, d_\ell\} \cup \{z_1, \dots, z_u\}$ for some $\ell$, where the sets in each union are disjoint, we have $\ell = t+r-u$.
Then, we define the sequence of candidates $b_1, \dots, b_{t+r}$ by letting $b_i = d_i$ for $i \in [t+r-u]$ and $b_i = z_{i-(t+r-u)}$ for $i \in \{t+r-u+1, \dots, t+r\}$. 

Finally, we define the transition committees as $W_i = (W' \cup \{b_1, \dots,b_i\}) \setminus \{a_1, \dots a_i\} $ for all $i \in \{0,\dots,t+r\}$.
Observe that $W_0 = W'$ and $W_{t+r} = \widehat{W}$.
We claim that all of the committees $W_i$ satisfy JR.
For $i\in\{0,\ldots,t\}$, since $W\subseteq W_i$ and $W$ satisfies JR, so does $W_i$.
Now, let $i\in\{t+1,\dots,t+r\}$.
We have $W_{i-t}^X = (W \cup \{d_1, \dots, d_s\}) \setminus \{c_1, \dots, c_{i-t}\}$ for some $s$.
Since $|W| = k-t$, as we showed earlier, it holds that $s\le (i-t)+t = i$.
Note that $W_i = (W\cup\{b_1,\dots,b_i\})\setminus\{c_1,\dots,c_{i-t}\}$.
Hence, $W_{i-t}^X \subseteq W_i$. 
Since $W_{i-t}^X$ satisfies JR, so does $W_i$.
Thus, $W$ and $W^X_\text{aff}$ are connected in $\JR(A,k)$.

It remains to show that  $W^X_\text{aff}$ is affordable. 
For this, the voters can first pay for the candidates in $X^*$ according to the payment system that witness the affordability of $X^*$. 
Then, each candidate $d_i$ added during the process was at the time approved by at least $\frac{n}{k}$ voters who approve no other candidate from the subcommittee and hence still had full budget of $\frac{k}{n}$ each available to pay for $d_i$. 
This concludes the proof. 
\end{proof}

We are now ready to show that any two affordable subcommittees satisfying JR are reachable from each other. 
\begin{restatable}{proposition}{affordconnect}
For any instance $(A,k)$, the set of affordable subcommittees satisfying JR is connected in $\JR(A,k)$.
\label{Thm:PayToWin}
\end{restatable}
\begin{proof}
We proceed by strong induction on the intersection size of the affordable committees. 
Fix any $s\in\{0,1,\dots, k\}$. 
Assume that all pairs of affordable subcommittees satisfying JR and intersecting in strictly more than $s$ candidates are connected in $\JR(A,k)$. 
We show that the same holds for all pairs of affordable subcommittees satisfying JR with intersection size $s$ (we will not treat the ``base case'' $s = k$ separately). 
Let $W_\text{aff}, W_\text{aff}'$ be affordable subcommittees satisfying JR with $\lvert W_\text{aff}\cap W_\text{aff}' \rvert = s$.
Denote the sizes of $W_\text{aff}$ and $W'_\text{aff}$ by $s+r$ and $s+t$, respectively, where $s,r,t\ge 0$. 

First, we consider the main case where $r,t>0$.
Let $X = W_\text{aff}\cap W'_\text{aff}$, and take any candidate $c'\in W'_\text{aff}\setminus X$.
Our goal is to add $c'$ to $W_\text{aff}$ and apply \Cref{Lemma:SubCommitteeExtension}.

If $|W_\text{aff}| \leq k-1$, then set $W = W_\text{aff}\cup \{c'\}$. 
Observe that $W$ meets the requirements of \Cref{Lemma:SubCommitteeExtension}, as $\lvert N_{W} \rvert \geq \lvert N_{W_\text{aff}} \rvert \geq (\lvert W
\rvert -1 )\frac{n}{k}$ by \Cref{Lemma:SubcommitteesAlgorithm}, and $W$ satisfies JR as a superset of $W_\text{aff}$.

Otherwise, $|W_\text{aff}| = k$, and so $\lvert N_{W_\text{aff}} \rvert = n$ by \Cref{Lemma:SubcommitteesAlgorithm}.
As $X \cup \{c'\} \subseteq W'_\text{aff}$ is an affordable subcommittee, we know from \Cref{Lemma:SubcommitteesAlgorithm} that $\lvert N_{X\cup\{c'\}}\rvert \ge (s+1)\cdot\frac{n}{k}$. 
Therefore, $\lvert N_{W_\text{aff}} \setminus N_{X\cup\{c'\}} \rvert \le n - (s+1)\cdot\frac{n}{k} = (k - s - 1)\cdot\frac{n}{k}$. 
As $W_\text{aff} \setminus X$ contains $k - s$ candidates, by the pigeonhole principle, there must exist a candidate $c \in W_\text{aff} \setminus X$ who is the only approved candidate in $W_\text{aff} \setminus X$ for fewer than $\frac{n}{k}$ voters from $N_{W_\text{aff}} \setminus N_{X\cup\{c'\}}$. 
By definition, every voter in $N_{X\cup\{c'\}}$ approves some candidate in $X\cup\{c'\}\subseteq (W_\text{aff}\cup\{c'\})\setminus\{c\}$.
Hence, the committee $ (W_\text{aff} \cup \{c'\}) \setminus \{c\}$ is approved by more than $(k-1)\cdot\frac{n}{k}$ voters.
By setting $W = (W_\text{aff}\cup \{c'\})\setminus \{c\}$, we get 
$\lvert N_{W} \rvert> \frac{n}{k}\cdot(k-1) = \frac{n}{k}\cdot(\lvert W \rvert - 1)$. 
This inequality also implies that fewer than $\frac{n}{k}$ voters are not represented by $W$, which means that $W$ satisfies JR. 
Hence, $W$ meets the requirements of \Cref{Lemma:SubCommitteeExtension} in this case too.
 
Applying \Cref{Lemma:SubCommitteeExtension} with $W$, $W'_\text{aff}$, and $X\cup \{c'\}$ yields an affordable subcommittee $W^{X\cup \{c'\}}_\text{aff}$ satisfying JR such that $X\cup \{c'\} \subseteq W^{X\cup \{c'\}}_\text{aff}$ and, moreover, $W$ and $W^{X\cup \{c'\}}_\text{aff}$ are connected in $\JR(A,k)$. 
Extending the chain of connections, $W^{X\cup \{c'\}}_\text{aff}$ and $W'_\text{aff}$ are connected in $\JR(A,k)$, since these two affordable subcommittees share all $s+1$ candidates in ${X\cup \{c'\}}$, so we can apply our induction assumption. 
Clearly, $W_\text{aff}$ and $W$ are connected in $\JR(A,k)$ by construction of $W$, as the two subcommittees are ``neighboring''. 
In summary, there is a chain of connections from $W_\text{aff}$ to $W$ to $W^{X\cup \{c'\}}_\text{aff}$ to $W_\text{aff}'$, which concludes the main case of the proof. 

Finally, we handle the edge case where at least one of $r$ and $t$ is zero.
If both $r$ and $t$ are zero, then $W_\text{aff} = W_\text{aff}'$, so the two committees are trivially connected. 
If only one of them is equal to zero, assume without loss of generality that $r = 0$.
Then, any superset of $W_\text{aff}$ satisfies JR, so we can replace the unneeded candidates one by one with candidates from $W_\text{aff}'\setminus W_\text{aff}$ without violating JR. \qedhere

\end{proof}

Several well-known proportional rules only select supersets of affordable subcommittees satisfying JR. 
Hence, we immediately obtain from \Cref{Thm:PayToWin} that the choice sets of such rules are always connected via JR committees.  Further, the choice sets of different rules are interconnected, e.g., we can transition from any MES committee to any seqCCAV committee without violating JR. 
While PAV and CCAV do not necessarily select supersets of affordable committees satisfying JR, we can nevertheless show that their choice sets are connected to a superset of an affordable subcommittee satisfying JR. 

\begin{restatable}{theorem}{corrules}\label{Cor:JRAllRulesConnected}
    Consider the set of rules
    \begin{align*}
        R = \{\text{MES, } &\text{seqCCAV, CCAV, PAV,}\\ &\text{GJCR, GreedyEJR, seqPhragm\'en}\}. 
    \end{align*}
    For any rules $f,f' \in R$ (possibly $f=f'$) and any instance $(A,k)$, any committee from $f(A,k)$ is connected to any other committee from $f'(A,k)$ in $\JR(A,k)$.
\end{restatable}
We divide the proof of this theorem into two parts: we first show the statement for all rules except PAV and then give a proof for PAV separately.

\begin{lemma}
    For any rule $f \in R\setminus\{\text{PAV}\}$ and any instance $(A,k)$, any committee from $f(A,k)$ is connected to some affordable subcommittee satisfying JR in $\JR(A,k)$.
    \label{lemma:everythingIsAffordable}
\end{lemma}
\begin{proof}
    For MES, GJCR, and GreedyEJR, it was already discussed by \citet{BrPe24a} that these rules always return supersets of affordable subcommittees satisfying JR.\footnote{\citet{BrPe24a} used the term ``committee'' for what we refer to as ``subcommittee'', that is, a committee in their paper may have size less than $k$.} 
    
    For a committee $W$ chosen by seqPhragm\'en, consider the subcommittee $W' \subseteq W$ chosen at time $t = \frac kn$; since the total budget across all voters up to this time is $\frac{k}{n}\cdot n = k$ and each candidate costs $1$, the rule has not terminated before this time. 
    As the family of payment functions $(p_i)_{v_i\in N}$, we take all the payments made up to this time.
    By definition of the rule, $(p_i)_{v_i\in N}$ satisfies all payment vector constraints: each voter has spent at most $\frac kn = t$ of her budget, and each bought candidate receives a total payment of $1$ from the voters.
    In addition, $W'$ satisfies JR, since each $1$-cohesive group has a total budget of $1$ available, so at least one voter in the group must be involved in a purchase. 
    This voter therefore approves at least one candidate in $W'$.
    
    Next, we show that seqCCAV always selects a superset of an affordable subcommittee that satisfies JR. 
    Let $W = \{c_1, \dots, c_k\}$ be a committee chosen by seqCCAV in this order of candidates. 
    By definition of seqCCAV, it holds that $\lvert N_{\{c_1, \dots, c_i \}} \setminus N_{\{c_1, \dots, c_{i-1} \}} \rvert \geq  \lvert N_{\{c_1, \dots, c_j \}} \setminus N_{\{c_1, \dots, c_{j-1} \}} \rvert$   whenever $1\le i < j \le k$.
    We consider two cases.
    
    \textbf{Case 1:} $\lvert N_{\{c_1, \dots, c_k \}} \setminus N_{\{c_1, \dots, c_{k-1} \}} \rvert = \frac{n}{k}$. 
    This implies that $N_W = N$. 
    Then, $W$ is affordable, since for every candidate $c \in W$ there are exactly $\frac{n}{k}$ voters that only approve $c$ in the committee $W$.
    
    \textbf{Case 2:} $\lvert N_{\{c_1, \dots, c_k \}} \setminus N_{\{c_1, \dots, c_{k-1} \}} \rvert < \frac{n}{k}$. 
    Then, there exists $s \in [k]$ such that $\lvert N_{\{c_1, \dots, c_i \}} \setminus N_{\{c_1, \dots, c_{i-1} \}} \rvert \geq \frac{n}{k}$ for all $i \in [s-1]$ and $\lvert N_{\{c_1, \dots, c_i \}} \setminus N_{\{c_1, \dots, c_{i-1} \}} \rvert < \frac{n}{k}$ for all $i \geq s$.
    We claim that the subcommittee $W' \coloneqq \{c_1, \dots, c_{s-1}\}$ is affordable and satisfies JR. 
    First, $W'$ is affordable because whenever a candidate $c_i$ is added to $W'$, there exist at least $\frac{n}{k}$ voters that approve this candidate but no candidate in $\{c_1, \dots, c_{i-1}\}$. 
    Those voters can pay for $c_i$, since they have not paid for any candidate in  $\{c_1, \dots, c_{i-1}\}$.
    Furthermore, $W'$ satisfies JR because there does not exist any candidate $c \in C$ such that $\lvert N_{W' \cup \{c\}} \setminus N_{W'} \rvert \geq \frac{n}{k}$.
    That is, for any $1$-cohesive group, at least one voter in the group approves some candidate in $W'$.

\begin{table}[tbh]
    \centering
    \begin{tabular}{@{}lccccccc@{}}
    \toprule
      & $v_1$ & $v_2$ & $v_3$ & $v_4$ & $v_5$ & $v_6$ & $v_7$ \\ \midrule
    $c_1$ & $\times$ & $\times$ &  &        &        &        \\
    $c_2$ &  &  & $\times$ & $\times$       &        &        \\
    $c_3$ &  & &  &  & $\times$ & $\times$ \\
    $c_4$ & $\times$ &  & $\times$ &  & $\times$ &  \\
\bottomrule
\end{tabular}
\caption{The instance $(A,k)$ described above with $k = 3$ is an example where $\CCAV(A,k) =  \{c_1,c_2,c_3\}$ does not contain an affordable subcommittee that satisfies JR. 
Indeed, since $v_1,v_3,v_5$ form a cohesive group (with $c_4$), a subcommittee satisfying JR must contain at least one candidate. 
However, every candidate in $\{c_1,c_2,c_3\}$ is approved by only two voters and therefore not affordable, since each voter only has a budget of $\frac{3}{7}$.}
\label{Example:CCAV}
\end{table}

    CCAV does not always select a superset of an affordable subcommittee that satisfies JR (see Table \ref{Example:CCAV}). 
    However, we can still show that it is connected to such a subcommittee. 
    Let $W$ be a committee returned by CCAV. 
    Consider an ordering $c_1, \dots, c_k$ of the candidates in $W$  such that $c_s \in \argmin_{c \in W \setminus  \{c_1, \dots, c_{s-1}\}} \lvert N_{c} \setminus  N_{W \setminus \{c_1, \dots, c_{s-1},c \}} \rvert $ for every $s \in [k]$, that is, $c_s$ is a candidate in $W\setminus\{c_1,\dots,c_{s-1}\}$ that is uniquely approved by the least number of voters.
    If $\lvert N_{c_1} \setminus  N_{W\setminus\{c_1\}} \rvert \geq \frac{n}{k}$, then every candidate in $W$ is uniquely approved by $\frac{n}{k}$ voters, so the committee $W$ itself is affordable. 
    Therefore, assume in the following that $\lvert N_{c_1} \setminus  N_{W\setminus\{c_1\}} \rvert < \frac{n}{k}$.
    Let $s\in [k]$ be the largest index such that $\lvert N_{c_s} \setminus N_{W \setminus \{c_1, \dots, c_{s}\}} \rvert  < \frac{n}{k}$. 
    By definition of $s$, each of the candidates $c_{s+1}, \dots, c_k$ is uniquely approved by at least $\frac{n}{k}$ voters in the subcommittee $\{c_{s+1}, \dots, c_k\}$. 
    (Note that if $s =k$, this set is empty.)
    We remove the candidates $c_1,\dots,c_s$ from $W$ in this order.
    After each removal, if a JR violation occurs, we add a candidate witnessing the violation into the subcommittee, and repeat until the subcommittee satisfies JR. 

    For each $i\in [s]$, we claim that if $i$ candidates have been removed, then at most $i-1$ candidates have been added.
    Assume for contradiction that the candidates $c_1, \dots, c_i$ are removed and the candidates $d_1, \dots, d_i$ are added.
    By definition of $s$, it holds that $\lvert N_{c_j} \setminus N_{W \setminus \{c_1,\dots, c_j\}}\rvert < \frac{n}{k}$ for each $j \in [i]$.
    Moreover, each candidate $d_j$ is approved by at least $\frac nk$ voters that approve no other candidate at the time $d_j$ is added, and the subcommittee at that time is a superset of the subcommittee $(W\cup\{d_1,\dots,d_{j-1}\})\setminus\{c_1,\dots,c_i\}$.
    Hence, we have

\begin{align*}
&\lvert N_{(W \cup \{d_1, \dots, d_i\})\setminus \{c_1,\dots, c_i\}} \rvert \\
&= \lvert N_{\{d_1,\dots, d_i\}}\setminus N_{W\setminus\{c_1,\dots, c_i\}}\rvert +  \lvert N_{W\setminus\{c_1,\dots, c_i\}} \rvert \\
&= (\lvert N_{d_1}\setminus N_{W\setminus\{c_1,\dots, c_i\}}\rvert + \lvert N_{d_2}\setminus N_{(W \cup \{d_1\}) \setminus\{c_1,\dots, c_i\}}\rvert \\
&\quad\quad+ \dots + \lvert N_{d_i}\setminus N_{(W \cup \{d_1,d_2,\dots, d_{i-1}\}) \setminus\{c_1,\dots, c_i\}}\rvert) \\
&\quad+ \lvert N_{W\setminus\{c_1,\dots, c_i\}} \rvert \\
&= (\lvert N_{d_1}\setminus N_{W\setminus\{c_1,\dots, c_i\}}\rvert + \lvert N_{d_2}\setminus N_{(W \cup \{d_1\}) \setminus\{c_1,\dots, c_i\}}\rvert \\
&\quad\quad+ \dots + \lvert N_{d_i}\setminus N_{(W \cup \{d_1,d_2,\dots, d_{i-1}\}) \setminus\{c_1,\dots, c_i\}}\rvert) \\
&\quad+ (\lvert N_W \rvert - \lvert N_{c_1} \setminus N_{W\setminus\{c_1\}}  \rvert - \lvert N_{c_2} \setminus N_{W\setminus\{c_1,c_2\}} \rvert \\
&\quad\quad- \dots - \lvert N_{c_i} \setminus N_{W\setminus\{c_1, \dots, c_i\}} \rvert) \\
&> i \cdot \frac{n}{k} + \left(\lvert N_W \rvert - i \cdot \frac{n}{k} \right)  \\
&= \lvert N_W \rvert.
\end{align*}
This contradicts the assumption that $W$ is returned by CCAV and establishes the claim.

    At the end of this process, we are left with the candidates $d_1,\dots,d_r$ that we added (for some $r < s$), along with $c_{s+1}, \dots, c_k$.
    We show that the resulting subcommittee is affordable.
    To see this, recall that by definition of $s$, each of the candidates $c_{s+1}, \dots,c_k$ is uniquely approved by at least $\frac{n}{k}$ voters with respect to the subcommittee $\{c_{s+1}, \dots,c_k\}$; these voters have enough budget to pay for these candidates. 
    Moreover, every candidate $d_j$ added during the process is approved by at least $\frac{n}{k}$ voters who approve no other candidate from the subcommittee at the time, and this subcommittee is a superset of $\{c_{s+1},\dots,c_k,d_1,\dots,d_{j-1}\}$. 
    Hence, these voters have enough budget to pay for $d_j$. 
    Therefore, the subcommittee $W_{\text{aff}} \coloneqq \{c_{s+1}, \dots, c_k\} \cup \{d_1, \dots, d_r\}$ is affordable. 
    Furthermore, $W_{\text{aff}}$ satisfies JR, as we have added candidates $d_j$ until no more JR violations exist.
    
    It remains to show that $W$ is connected to  $W_{\text{aff}}$.
    Note that $W_{\text{aff}}$ has size $k - s + r$.
    Let $W' \supseteq W_{\text{aff}}$ be any committee with $W' \setminus W_{\text{aff}} = \{y_1, \dots, y_{s-r}\}$. 
    We define a sequence of transitions $W_0, \dots, W_s$ with $W_0 = W$, $W_i = (W \cup \{d_1, \dots, d_i\}) \setminus \{c_1, \dots, c_i\}$ for $i \in [r]$, 
    and $W_i = (W \cup \{d_1, \dots, d_r\} \cup \{y_1, \dots, y_{i-r}\}) \setminus \{c_1, \dots, c_i\}$ for $r < i \leq s$. 
    As we showed earlier, for each $i \in [r]$, after removing $c_1, \dots, c_i$, it is sufficient to add $d_1, \dots, d_{i-1}$ in order to satisfy JR. 
    Hence, for every $i \in [r]$, the committee $W_i$ satisfies JR.
    Furthermore, for $ r < i \leq s$, since $W_{\text{aff}}$ satisfies JR and 
    \begin{align*}
    W_{\text{aff}} 
    &= \{c_{s+1}, \dots,c_k\} \cup \{d_1, \dots, d_r\} \\
    &\subseteq \{c_{i+1}, \dots,c_k\} \cup \{d_1, \dots, d_r\} \cup \{y_1, \dots, y_{i-r} \} \\
    &= W_i,
    \end{align*}
    we can conclude that $W_i$ satisfies JR.
    Therefore, $W$ and $W_{\text{aff}}$ are connected in $\JR(A,k)$.
\end{proof}

We next prove an analogous statement for PAV. 
To this end, we first formulate a lemma which guarantees that we can always remove some candidate with a low PAV score contribution. 
Recall that $\Delta(W,c)$ denotes the marginal contribution of $c$ to the PAV score of $W$, that is, $\PAV(W) - \PAV(W\setminus\{c\})$.

\begin{lemma}
        For any subcommittee $W$, if $\lvert N_{W} \rvert < \frac{n}{k}\cdot(\lvert W \rvert -1)$, then a candidate $c\in W$ that has the minimum marginal contribution to the PAV score of $W$ satisfies  $\Delta(W,c) <  \frac{n}{k}$.
        \label{PAVLemma}
    \end{lemma}
\begin{proof}
    Assume that $\lvert N_{W} \rvert < \frac{n}{k}\cdot(\lvert W \rvert -1)$, and note that $W$ cannot be empty.
    We have
    \begin{align*}
    \sum_{c\in W}\Delta(W,c) 
    &= \sum_{c\in W} \sum_{v \in N_c} \frac{1}{\lvert A_v \cap W \rvert} \\
    &= \sum_{v \in N_W} \sum_{c \in W\cap A_v} \frac{1}{\lvert A_v \cap W \rvert} \\
    &= \sum_{v \in N_W} 1 
    = \lvert N_W \rvert 
    < \frac{n}{k} \cdot (\lvert W \rvert -1).
    \end{align*} 
    By the pigeonhole principle, there exists a candidate $c \in W$ such that  $\Delta(W,c) < \frac{n}{k} \cdot\frac{\lvert W \rvert -1}{\lvert W\rvert}< \frac{n}{k}$.
\end{proof} 

We now establish the statement for PAV.

\begin{lemma}
    For any instance $(A,k)$, any committee that is returned by PAV is connected to some affordable subcommittee satisfying JR in $\JR(A,k)$.
    \label{Lemma:PavAff}
\end{lemma}
\begin{proof}
    Let $W_{\PAV}$ be any committee returned by PAV. Our goal is to connect $W_{\PAV}$ to some subcommittee $W'$ with $\lvert N_{W'} \rvert \geq \frac{n}{k}(\lvert W' \rvert -1)$ that satisfies JR. 
    Then, we can apply \Cref{Lemma:SubCommitteeExtension} by choosing $W'_{\text{aff}}$ and $X$ to be the empty set---the lemma guarantees the existence of an affordable subcommittee $W_{\text{aff}}^X$ satisfying JR that is connected to $W$ in $\JR(A,k)$.

    If $W_{\PAV}$ itself satisfies $\lvert N_{W_{\PAV}} \rvert \geq \frac{n}{k}(\lvert W_{\PAV} \rvert -1)$, we are done since $W_{\PAV}$ also satisfies JR. 
    Therefore, assume in the following that $\lvert N_{W_{\PAV}} \rvert < \frac{n}{k}(\lvert W_{\PAV} \rvert -1) = \frac{n}{k}(k-1)$.   
    We will determine a suitable subcommittee $W' = (W \setminus \{c_1,\dots, c_{s}\}) \cup \{d_1,\dots,d_r\}$ through the following process.
    
    \textbf{Candidate removals:}
    We initialize $W' = W_{\PAV}$ and
    iteratively remove candidates $c_1,c_2,\dots$ from $W$ in such a way that, for each $i$, the marginal contribution of $c_i$ to the PAV score of $W_{\PAV} \setminus \{c_1,\dots,c_{i-1}\}$ is minimal. 
    We stop the removal process after the smallest index $i$ such that  
    $\lvert N_{W_{\PAV} \setminus \{c_1,\dots,c_i\}} \rvert \geq (k-i-1)\cdot\frac{n}{k}$, and call this smallest index $s$. 
    Note that for $i = k$, the inequality trivially holds, so $s$ is well-defined. 
    Moreover, $s \geq 1$, as we assumed that $W_{\PAV}$ itself does not satisfy this condition. 

    \textbf{Candidate additions:}
    After each removal, if a JR violation occurs, we add a candidate witnessing the violation into the subcommittee, and repeat until the subcommittee satisfies JR. 
    We claim that for any set $W' = (W_{\PAV}\setminus \{c_1,\dots,c_i\}) \cup \{d_1,\dots,d_j\}$ created during the process such that JR is restored in $W'$, it holds that $i > j$. To prove this claim, we first show by induction on $i$ that, for every $i\in[s]$, it holds that $\PAV(W_{\PAV} \setminus \{c_1, \dots, c_i\})   > \PAV(W_{\PAV}) - i\cdot\frac{n}{k}$.

    For the base case $i = 1$, since $\lvert N_{W_{\PAV}} \rvert < \frac{n}{k}(\lvert W_{\PAV} \rvert  - 1)$, \Cref{PAVLemma} implies that the marginal contribution of $c_1$ to the PAV score of $W_{\PAV}$ is less than $\frac{n}{k}$.
    Hence, $\PAV(W_{\PAV} \setminus \{c_1\})   > \PAV(W_{\PAV}) - \frac{n}{k}$.
    
    For the induction step, assume that the statement holds for some $i\in [s-1]$.
    By definition of $s$, it must hold that $\lvert N_{W_{\PAV}\setminus \{c_1,\dots,c_i\}} \rvert < \frac{n}{k}(k-i-1) = \frac{n}{k}(\lvert W_{\PAV}\setminus \{c_1,\dots,c_i\} \rvert -1)$. 
    Therefore, \Cref{PAVLemma} implies that the marginal contribution of $c_{i+1}$ to the PAV score of $W_{\PAV}\setminus \{c_1,\dots,c_i\}$ is less than $\frac{n}{k}$. 
    Hence, we have 
    \begin{align*}
    &\PAV(W_{\PAV} \setminus \{c_1, \dots, c_i,c_{i+1}\}) \\
    &> \PAV(W_{\PAV} \setminus \{c_1, \dots, c_i\}) - \frac{n}{k} \\
    &> \PAV(W_{\PAV}) - (i+1)\cdot\frac{n}{k},
    \end{align*}
    where the last inequality follows from the induction hypothesis.
    This concludes the induction.

    Next, we claim that for each $i\in[s]$, after removing a total of $i$ candidates, we need to add a total of fewer than $i$ candidates in order to restore JR.
    Assume for contradiction that, after removing $c_1,\dots,c_i$ and adding $d_1,\dots,d_{i-1}$, JR is not yet restored.
    Observe that for each $j\in [i-1]$, just before $d_j$ is added, the subcommittee is a superset of $\widehat{W} \coloneqq (W_{\PAV} \setminus \{c_1, \dots, c_i\}) \cup \{d_1,\dots,d_{j-1}\}$.
    Since $d_j$ is a witness of a JR violation, there exist $\frac{n}{k}$ voters who do not approve any candidate in $\widehat{W}$, so adding $d_j$ to $\widehat{W}$ increases its PAV score by at least $\frac{n}{k}$.
    This implies that
    \begin{align*}
    &\PAV((W_{\PAV} \setminus \{c_1, \dots, c_i\}) \cup \{d_1,\dots,d_{i-1}\}) \\
    & \ge \PAV(W_{\PAV} \setminus \{c_1, \dots, c_i\}) + (i-1)\cdot\frac{n}{k} \\
    &> \PAV(W_{\PAV}) - i\cdot\frac{n}{k} + (i-1)\cdot\frac{n}{k} \\
    &= \PAV(W_{\PAV}) - \frac{n}{k}.
    \end{align*}
    Now, since $(W_{\PAV} \setminus \{c_1, \dots, c_i\}) \cup \{d_1,\dots,d_{i-1}\}$ does not satisfy JR, adding a candidate that witnesses the JR violation increases the PAV score by at least $\frac{n}{k}$, thereby making it higher than the score of $W_{\PAV}$.
    This contradicts the assumption that $W_{\PAV}$ is chosen by PAV.

    \textbf{Construction of the transition committees:}
    Let $W' = (W_{\PAV} \setminus \{c_1,\dots, c_{s}\}) \cup \{d_1,\dots,d_r\}$ be the final subcommittee created by the procedure, and $W'' = W' \cup \{y_1, \dots, y_{s-r}\}$ be any committee containing $W'$.
    We now define a sequence $W_0, \dots, W_s$ with $W_0 = W_{\PAV}$, $W_s = W''$,  $W_i = (W_{\PAV}  \setminus \{c_1,\dots,c_i\}) \cup \{d_1,\dots,d_{i}\}$ for $i\in [r]$, and $W_i = (W_{\PAV}  \setminus \{c_1,\dots,c_i\}) \cup \{d_1,\dots,d_{r}\} \cup \{y_1,\dots,y_{i-r}\}$ for $r < i < s$. 
   We have shown above that for $i\in [r]$, after removing $c_1, \dots, c_i$, it is sufficient to add $d_1, \dots, d_{i-1}$ in order to restore JR. 
   Therefore, for $i \leq r$, the committee $W_i$ satisfies JR.
    For $ r < i \leq s$, since $W'$ satisfies JR and
    \begin{align*}
    W' &= \{c_{s+1}, \dots,c_k\} \cup \{d_1, \dots, d_r\} \\
    &\subseteq \{c_{i+1}, \dots,c_k\} \cup \{d_1, \dots, d_r\} \cup \{y_1, \dots, y_{i-r} \} \\
    &= W_i,
    \end{align*}
    we can conclude that $W_i$ again satisfies JR.
    Therefore, $W_{\PAV}$ and $W'$ are connected in $\JR(A,k)$.

    It remains to show that the subcommittee $W'$ satisfies $\lvert N_{W'}\rvert \geq \frac{n}{k}\cdot(\lvert W' \rvert -1)$. 
    We have 
    
     \begin{align*}
    &\lvert N_{W'} \rvert \\
    &= \lvert N_{\{d_1, \dots,d_r \}} \setminus N_{W_{\PAV}  \setminus \{c_1,\dots,c_s\}} \rvert + \lvert N_{W_{\PAV}  \setminus \{c_1,\dots,c_s\}} \rvert  \\
    &= (\lvert N_{d_1}\setminus N_{W_{\PAV}\setminus\{c_1,\dots, c_s\}}\rvert \\
    &\quad+ \lvert N_{d_2}\setminus N_{(W_{\PAV} \cup \{d_1\}) \setminus\{c_1,\dots, c_s\}}\rvert \\
&\quad+ \dots + \lvert N_{d_r}\setminus N_{(W_{\PAV} \cup \{d_1,d_2,\dots, d_{r-1}\}) \setminus
\{c_1,\dots, c_s\}}\rvert) \\
&\quad+\lvert N_{W_{\PAV}  \setminus \{c_1,\dots,c_s\}} \rvert \\
    & \geq r\cdot\frac{n}{k} + \lvert N_{W_{\PAV}  \setminus \{c_1,\dots,c_s\}} \rvert \\
    & \geq r\cdot\frac{n}{k} + (k - s  -1) \cdot\frac{n}{k} \\
 &= (r+k-s-1)\cdot \frac{n}{k} \\
 &= (\lvert W' \rvert - 1)   \cdot\frac{n}{k},
\end{align*} 
where we use the definition of $s$ for the last inequality.
This completes the proof.
\end{proof}

Finally, we have all the ingredients necessary for proving \Cref{Cor:JRAllRulesConnected}.

\begin{proof}[Proof of \Cref{Cor:JRAllRulesConnected}] 
    Let $W \in f(A,k)$ and $W' \in f'(A,k)$ be two committees returned by $f,f' \in R$, respectively. 
    By Lemmas \ref{lemma:everythingIsAffordable} and \ref{Lemma:PavAff}, there exists an affordable subcommittee $W_a$ satisfying JR such that $W$ and $W_a$ are connected in $\JR(A,k)$; similarly, there exists an affordable subcommittee $W'_a$ satisfying JR such that $W'$ and $W'_a$ are connected in $\JR(A,k)$. 
    Furthermore, by Proposition~\ref{Thm:PayToWin}, $W'_a$ and $W_a$ are connected in $\JR(A,k)$. 
    It follows that $W$ and $W'$ are also connected in $\JR(A,k)$.
\end{proof}

Hence, while the set of JR committees can be rather disconnected (\Cref{thm:isolated_jr}), the committees chosen by several important rules all lie in the same connected component within the set of committees satisfying JR.

For EJR and EJR+, we will establish a weaker statement that all committees chosen by 
PAV, MES, and GJCR are not isolated in EJR+ and EJR.
We first prove two lemmas to handle special cases in which a committee is trivially not isolated. 

\begin{lemma}
    For any instance $(A,k)$  and any $\Pi \in \{\EJR, \EJRP\}$, let $W$ be a committee that satisfies $\Pi$. Further, suppose that there exists a candidate $c \in W$ such that $W \setminus \{c\}$ also satisfies $\Pi$. 
    Then, $W$ is not isolated in $\Pi(A,k)$.
    \label{Lemma:StrictSubcommitteeSatisfiesX}
\end{lemma}

\begin{proof}
    If $C = W$, then $W$ is the only committee, so it is not isolated by definition.  
    Otherwise, there exists a candidate $c' \in C \setminus W$.
    Since we assume that $W \setminus \{c\}$ satisfies $\Pi$ for some $c \in W$, it holds that $(W \cup \{c'\}) \setminus \{c\}$ also satisfies $\Pi$. 
    The claim then follows from the observation that $W'$ can be reached from $W$ by replacing a single candidate. 
\end{proof}

\begin{lemma}
     For any instance $(A,k)$  and any $\Pi \in \{\EJR, \EJRP\}$, let $W$ be a committee that satisfies $\Pi$. Furthermore, let $C_a\subseteq C$ be the set of candidates that is approved by at least one voter. 
     If $\lvert C_a\rvert \leq k$, then $W$ is not isolated in $\Pi(A,k)$.
    \label{Lemma:atMostKRealCandidates}
\end{lemma}
\begin{proof}
    Assume that $\lvert C_a\rvert\leq k$. 
    If $C = W$, then $W$ is the only committee, so it is not isolated by definition. 
    Suppose therefore that $\lvert C\rvert > k$.    
    If there exists a candidate $c\in W$ that is not approved by any voter, take any candidate $c' \in C \setminus  W$. 
    Then, $(W \cup \{c'\})\setminus \{c\}$ satisfies $\Pi$, and so $W$ is not isolated. 
    
    Assume from now on that every candidate in $W$ is approved by at least one voter.
    This means that $W\subseteq C_a$.
    Since $|W| = k \ge |C_a|$, it must be that $C_a = W$.
    If there is no committee other than $W$ that satisfies $\Pi$, then $W$ is not isolated by definition.
    Hence, let $W' \neq W$ be a committee satisfying $\Pi$. 
    Since $W'\neq W$, there exists a candidate $c \in W \setminus W'$ and a candidate $c' \in W' \setminus W$. 
    We claim that $(W \cup \{c'\}) \setminus \{c\}$ satisfies $\Pi$. To this end, observe that $W' \cap C_a \subseteq W \setminus \{c\}$. 
    This implies that for each voter, the set of candidates that the voter approves in $W \setminus \{c\}$ is a superset of the corresponding set for $W'$. 
    Since $W'$ satisfies $\Pi$, so does $(W \cup \{c'\}) \setminus \{c\}$. It follows that $W$ is not isolated.
\end{proof}

We now proceed to the proof of non-isolation for PAV.

\begin{lemma}
    Let $W$ be a committee returned by PAV such that for each $c\in W$, the subcommittee $W \setminus \{c\}$ violates EJR+.
    Further, suppose that there exists a candidate in $C \setminus W$ that is approved by at least one voter. 
    Then, there exists a committee $W'$ satisfying EJR+ such that $d(W,W') = 1$. 
    \label{prop:PAVnotIsolatedEJRPlus}
\end{lemma}
\begin{proof}
    Let $c^*\in C\setminus W$ be a candidate approved by some voter $v^*$.
    Recall that $\Delta(W,c) \coloneqq \PAV(W) - \PAV(W\setminus\{c\})$.
    We first show that for every $c\in W$, it holds that $\Delta(W,c) \ge \frac nk$. 
    Consider any $c\in W$.
    Since $W\setminus \{c\}$ violates EJR+, for some $\ell$, there exists a set $N'\subseteq N$ of at least $\ell \cdot\frac nk$ voters such that every voter $v\in N'$ has $\lvert A_v\cap (W\setminus \{c\})\rvert < \ell$ and all of these voters approve the same candidate $c' \notin W\setminus\{c\}$.
    This means that 
    \begin{align*}
    \Delta((W \cup \{c'\}) \setminus\{c\}, c') \ge \frac 1 \ell \cdot \ell \cdot\frac nk = \frac nk.
    \end{align*}
    Hence, the PAV score of the committee $(W \cup\{c'\})\setminus\{c\} $ is at least $\frac{n}{k}$ larger than that of the subcommittee $W\setminus\{c\}$. 
    Since $W$ maximizes the PAV score among all committees, we must have
    \begin{align*}
    \Delta(W,c) + \PAV(W \setminus \{c\}) 
    &= \PAV(W) \\
    &\geq \PAV((W \cup \{c'\}) \setminus \{c\})\\ 
    &\ge \frac{n}{k} + \PAV(W \setminus \{c\}), 
    \end{align*}
which implies that $\Delta(W,c) \ge \frac{n}{k}$.

Now, it holds that
    \begin{align*}
    n &= \sum_{c\in W}\frac{n}{k} \\
    &\le\sum_{c\in W}\Delta(W,c) \\
    &= \sum_{c\in W} \sum_{v \in N_c} \frac{1}{\lvert A_v \cap W \rvert} \\
    &= \sum_{v \in N_W} \sum_{c \in W\cap A_v} \frac{1}{\lvert A_v \cap W \rvert} 
    = \sum_{v \in N_W} 1 
    = \lvert N_W \rvert \le n.
    \end{align*} 
It follows that $N_W = N$ and $\Delta(W,c) = \frac{n}{k}$ for all $c\in W$.

    Let $c\in W$ be such that $c \in A_{v^*}$; such a candidate $c$ exists since $N_W = N$.
    Assume that the removal of $c$ from $W$ causes a violation of EJR+ due to some $\ell$-large set of voters $N'$, with all voters in $N'$ approving some $c'\not\in W\setminus\{c\}$. 
    We will show that the committee $W_{c'} \coloneqq (W\cup \{c'\})\setminus \{c\}$ is chosen by PAV and hence satisfies EJR+.
    As in earlier calculations, the marginal contribution of $c'$ to the PAV score of $W_{c'}$ is at least $\frac nk$. 
    Hence, the PAV score of $W_{c'}$ is at least that of $W$.
    Since $W$ is chosen by PAV, the scores of both committees must be equal, and $W_{c'}$ is chosen by PAV as well.
    If there exists such a candidate $c'\ne c$, then $W_{c'}$ is a committee that satisfies EJR+ and has distance $1$ from $W$, and we are done.
    
    Suppose in the following that there does not exist a candidate $c' \neq c$ such that the subcommittee $W \setminus \{c\}$ witnesses an EJR+ violation due to $c'$, so the violation is only due to $c$ itself.
    Then, for some $\ell$, the set of voters approving $c$ must be of size exactly $\ell\cdot\frac{n}{k}$, and each of these voters must approve exactly $\ell-1$ candidates from $W\setminus\{c\}$---otherwise, by similar calculations as earlier, the marginal contribution of $c$ to the PAV score of $W$ would be larger than $\frac{n}{k}$.
    Recall that voter $v^*$ approves both $c\in W$ and $c^*\in C\setminus W$.
    It follows that $W_{c^*} \coloneqq (W\cup \{c^*\})\setminus \{c\}$ satisfies EJR+, since the only potential EJR+ violation is due to the set of $\ell\cdot\frac{n}{k}$ voters approving $c$, and because of $v^*$, this set does not constitute an EJR+ violation. 
    Since $d(W,W_{c^*}) =1$, this concludes the proof.
\end{proof}

\begin{theorem}
\label{thm:PAV-non-isolation}
     For any instance $(A,k)$  and any $\Pi \in \{\EJR, \EJRP\}$, each committee returned by PAV is not isolated in $\Pi(A,k)$.
    \label{Corollary:PAVCorollary}
    \end{theorem}

\begin{proof}
    Let $W$ be any committee chosen by PAV.
    We may assume that for all $c \in W$ the subcommittee $W \setminus \{c\}$ violates $\Pi$, and there exists a candidate in $C\setminus W$ that is approved by at least one voter.    
    Indeed, otherwise we can apply \Cref{Lemma:StrictSubcommitteeSatisfiesX} or \ref{Lemma:atMostKRealCandidates} to conclude that $W$ is not isolated in $\Pi(A,k)$. 

    Suppose first that $\Pi = \EJRP$.
    By \Cref{prop:PAVnotIsolatedEJRPlus}, there exists a committee $W'$ satisfying EJR+ such that $d(W,W') = 1$. 
    Hence, $W$ is not isolated in $\EJRP(A,k)$.

    Suppose now that $\Pi = \EJR$.
    This means that for all $c \in W$, the subcommittee $W \setminus \{c\}$ violates EJR. 
    Since a violation of EJR is also a violation of EJR+, by \Cref{prop:PAVnotIsolatedEJRPlus} again, there exists a committee $W'$ satisfying EJR+ such that $d(W,W') = 1$. 
    Then, $W'$ also satisfies EJR, and therefore $W$ is not isolated in $\EJR(A,k)$.
\end{proof}

Next, we show the analogous statement for both MES and GJCR. 
Note that if a subcommittee $W'$ of size strictly less than $k$ satisfies EJR+ and is contained in some committee $W$, then $W$ is trivially not isolated in $\EJRP(A,k)$ and $\EJR(A,k)$. 
Thus, if GJCR or MES produces a subcommittee of size strictly less than $k$ and completes it arbitrarily, the non-isolation trivially holds.
Assume therefore that the rule produces a committee without the arbitrary completion phase.
We enumerate the candidates $c_1,\dots,c_k\in W$ according the order in which they are added to $W$.

For both rules, each voter is endowed with a budget of $\frac{k}{n}$, and each candidate added to the committee incurs a cost of~$1$ to the voters. 
Therefore, in order to pay for $k$ candidates in total, every voter has to spend all of her budget.
For a voter $v \in N$, let $W(v)$ denote the set of candidates that $v$ pays for.
Similarly, for a candidate $c$, let $N(c)$ denote the set of voters that pay for $c$.

To prove the non-isolation, we need one additional lemma.

\begin{lemma}\label{MESEJRPViolationLemma}
    Let $(c, N', \ell)$ be any witness of an EJR+ violation of a subcommittee $\{c_1, \dots, c_{x}\}$ during the execution of MES or GJCR, for some $x\in [k-1]$.
 Then, every voter in $N'$ still has a budget of at least $\frac{k}{\ell \cdot n}$ after the purchase of $c_x$. 
\end{lemma}

\begin{proof}
    Assume for the sake of contradiction that the claim does not hold, that is, some voter in~$N'$ spent more than $\frac{\ell-1}{\ell}\cdot\frac{k}{n}$.
    Note that every voter in $N'$ paid for at most  $\ell -1$ candidates in  $C_x \coloneqq \{c_1, \dots, c_{x}\}$. 
    
    For MES, we know that a voter in~$N'$ paid more than $\frac{k}{\ell \cdot n}$ for some candidate in $C_x$. 
    Let $c' \in C_x$ be the first such candidate across all such voters.  
    Before $c'$ is bought, every voter in $N'$ still had a budget of at least $\frac{k}{\ell \cdot n}$ to buy $c$.
    Hence, $c$ should have been bought instead of $c'$, a contradiction. 

    Similarly, for GJCR, we know that $(c, N', \ell)$ is also a witness of an EJR+ violation for any prefix of $C_x$. 
    This means that up to the point after the purchase of $c_x$, each voter paid at most $\frac{k}{\ell \cdot n}$ per candidate.
    It follows that every voter in $N'$ paid at most $\frac{\ell-1}{\ell}\cdot \frac{k}{n}$, a contradiction.
\end{proof}

We are now ready to show the non-isolation of MES and GJCR.

\begin{theorem}
     For any instance $(A,k)$  and any $\Pi \in \{\EJR, \EJRP\}$, each committee returned by MES and GJCR is not isolated in $\Pi(A,k)$.
    \label{Corollary:MESCorollary}
\end{theorem}

\begin{proof}
Let $W$ be any committee chosen by MES or GJCR. We may assume that there exists a candidate in $C\setminus W$ that is approved by at least one voter---indeed, otherwise we can apply \Cref{Lemma:atMostKRealCandidates} to conclude that $W$ is not isolated in $\Pi(A,k)$.     
We may also assume that for every $c\in W$, the subcommittee $W \setminus \{c\}$ suffers from an EJR+ violation, since otherwise we can apply \Cref{Lemma:StrictSubcommitteeSatisfiesX}.
Furthermore, as argued after the proof of \Cref{thm:PAV-non-isolation}, we may assume that every voter spends her entire budget $\frac{k}{n}$ on~$W$. 

Let $t\in[k]$ be the largest index for which there exists a candidate $c \in C \setminus W $ and a voter $v$ such that $c \in A_v $ and $v \in N(c_t)$. 
In the following, we will show that the
committee $W' \coloneqq  (W \cup \{c\}) \setminus \{c_t\}$, which has distance $1$ to $W$, satisfies EJR+, thereby establishing the non-isolation of $W$. 
Assume for the sake of contradiction that $W'$ violates EJR+, and let $(c', N', \ell)$ be a witness of this violation. 
As $c'$ is also a witness of an EJR+ violation for $\{c_1, \dots, c_{t-1}\}$, we know by \Cref{MESEJRPViolationLemma} that every voter in $N'$ must have a budget of at least $\frac{k}{\ell\cdot n}$ left before $c_t$ was bought. 
By definition of $t$, voter~$v$ did not buy any candidate in $\{c_{t+1}, \dots, c_k\}$. 
By definition of GJCR and MES, $v$ spent at most $\frac{k}{\ell\cdot n}$ on any candidate before $c_t$.
Moreover, $v$ spent all of her remaining money on $c_t$.
This amount is again at most $\frac{k}{\ell \cdot n}$, since $c'$ could be bought for a price of at most $\frac{k}{\ell \cdot n}$ at the time when $c_t$ got purchased.
Since $v$ also approves $c$, we have that $v$ approves at least $\ell$ candidates in $W'$. 
Therefore, $v\not\in N'$.

We claim that every voter in $N'$ pays for $c_t$. 
Assume for contradiction that this is not the case.
This means that $c' \neq c_t$, and there exists a voter $v' \in N'$ that has money left after the purchase of $c_t$. 
In order for $v'$ to spend her entire budget, there must exist a  candidate $c_r$ with $r > t$ that $v'$ buys. 
However, this contradicts the maximality of $t$, as $v'$ approves $c' \in C \setminus W$ and buys $c_r$. 
Hence, every voter in $N'$ pays for $c_t$.
Now, recall that each member of $N'$ has a budget of at least $\frac{k}{\ell \cdot n}$ left before $c_t$ was bought, and $v$ spent positively on $c_t$. 
Thus, each member of $N'$ spent less than $\frac{k}{\ell \cdot n}$ on $c_t$, and therefore still has budget left after buying $c_t$.
For each candidate $c_s \in \{c_{t+1},  \dots, c_k\}$, we know  by assumption that removing the candidate from~$W$ causes an EJR+ violation (with candidate $c_s$ itself, as the voters buying $c_s$ approve no candidates outside of $W$).  
Hence, for every such $c_s$, there exists a witness $(c_s, N'_s, \ell_s)$ of an EJR+ violation of $W \setminus \{c_s\}$. 

Recall the definition of a payment system (\Cref{def:affordable}).
Instead of considering the payment system of MES or GJCR for the candidates in $\{c_{t+1}, \dots, c_k\}$, for each $s\in\{t+1,\dots,k\}$, assume that each $v_i \in N'_s$ pays $\frac{1}{\lvert N'_s \rvert} \leq \frac{k}{\ell_s \cdot n}$ for $c_s$. 
For the candidates in $\{c_1, \dots, c_t\}$, we keep the payments as done by MES and GJCR.
We claim that this leads to a valid payment system for $W$, i.e., every voter pays at most $\frac{k}{n}$ in total and only for candidates that she approves, and each candidate costs exactly $1$ in total.
The fact that each candidate costs $1$ follows directly from our construction of the payment system.

A voter that does not pay for any candidate in $\{c_{t+1}, \dots, c_k\}$ pays at most $\frac{k}{n}$ in total, due to the validity of the payment systems of GJCR and MES.
Let
$v_i \in N$ be an arbitrary voter that pays for at least one candidate in $\{c_{t+1}, \dots, c_k\}$---let $c_r$ be one such candidate---and denote by
$q_i$ be the number of candidates in $W$ that voter $v_i$ pays for. 
For every candidate $c_s \in \{c_{t+1}, \dots,c_k\}$ that voter $v_i$ pays for in this payment system, she pays at most
$\frac{k}{q_i \cdot n}$. 
Indeed, for $s \in \{{t+1}, \dots, k\}$, in order for $v_i \in N'_s$ to hold, it must be that $\ell_s \geq q_i$: voter $v_i$ approves at least $q_i$ candidates in $W$, since she approves every candidate she pays for. 
Furthermore, in both MES and GJCR, she does not pay more than $\frac{k}{q_i \cdot n}$ for any candidate in $\{c_1, \dots, c_t\}$. 
To see this, recall that there exists an EJR+ violation for $W \setminus \{c_r\}$ due to some witness $(c_r,N'_r,\ell_r)$. This violation is therefore also an EJR+ violation for the subset $\{c_1, \dots, c_t\}$.
By Lemma \ref{MESEJRPViolationLemma}, every voter still had a budget of at least $\frac{k}{\ell_r \cdot n}$ after $c_t$ was purchased. 
Thus, the candidate $c_r$ could have been bought at the purchase of any candidate in $\{c_1, \dots, c_t\}$ for the maximum price of at most $\frac{k}{\ell_r \cdot n} \leq \frac{k}{q_i \cdot n}$, where the inequality $\ell_r \ge q_i$ holds for similar reasons as $\ell_s\ge q_i$ above.
Therefore, voter $v_i$ paid at most $\frac{k}{q_i \cdot n}$ for any candidate in $\{c_1, \dots, c_t\}$. 
Since voter $v_i$ does not pay more than $\frac{k}{q_i \cdot n}$ for any candidate, we conclude that $v_i$ pays at most $\frac{k}{n}$ in this payment system, and the payment system is valid.

Now, since there are $k$ candidates and each candidate costs~$1$ in this payment system, every voter spent her entire budget of $\frac{k}{n}$ on the candidates in $W$. 
Consider any voter $v_i \in N'$ paying for $q_i$ candidates in $W$. 
As observed earlier, $v_i$ still has some budget left after $c_t$ gets bought. 
Therefore, $v_i$ has to buy at least one candidate in $\{c_{t+1}, \dots,c_k\}$. 
From the previous paragraph, $v_i$ pays at most $\frac{k}{q_i \cdot n}$ for any candidate she purchases. 
Further, recall that for candidate $c_t$, she pays strictly less than $\frac{k}{\ell \cdot n}$.
Since $v_i \in N'$, it holds that $\ell \geq q_i$. 
Therefore, $v_i$ pays strictly less than  $\frac{k}{q_i\cdot n}$ for candidate $c_t$ and at most $\frac{k}{q_i \cdot n}$ for the other $q_i-1$ candidates she purchases.
It follows that $v_i$ did not spend her entire budget.
This yields the desired contradiction.    
\end{proof}

This leads us to our open question for this section: can we extend the non-isolation result to show that these committees are also connected to each other? 
\begin{openq}
    Are the committees returned by PAV, MES, and GJCR connected 
    in
    $\EJR(A,k)$ (resp., $\EJRP(A,k)$) for all instances $(A,k)$?    
\end{openq}

\section{Restricted Preference Domains}
\label{sec:restricted-domains}

Finally, we turn our attention to \textit{restricted domains}, 
in which the approval profiles are assumed to satisfy some structural constraints.
Specifically, we study two common preference domains, \emph{candidate interval} (CI) and \emph{voter interval} (VI) \citep{ElLa15a}. Both are frequently used in the approval-based multiwinner voting literature and often allow circumventing impossibilities that arise without domain restrictions \citep[see, e.g.,][]{PiSk22c, BIMP24a, DBW+24a}.

We start with the CI domain. An instance $(A, k)$ is in the CI domain if there exists an order $c_1, \dots, c_m$ of the candidates such that for every voter $v \in N$, the set $\{j \in [m] \mid c_j \in A_v\}$ forms an interval of consecutive integers. 
If an instance belongs to this domain, we show that the set of JR committees is connected in a ``shortest-path'' manner.
\begin{theorem}\label{thm:JR_CI_direct_connection}
    On the CI domain, the set of JR committees is connected.
    Moreover, any two JR committees $W,W'$ are connected by a path of JR committees of length $d(W,W')$. 
\end{theorem}
\begin{proof}
Let $W = \{d_1,\dots, d_k\}$ and $W' = \{e_1,\dots, e_k\}$ be two distinct JR committees, where the enumeration is according to the CI ordering.
Let $x,y\in[k]$ be the smallest indices such that $d_x\notin W'$ and $e_y\notin W$, and assume without loss of generality that $d_x$ comes before $e_y$ in the CI ordering.

We claim that the committee $W^* = (W\cup\{e_y\})\setminus\{d_x\}$ satisfies JR.
To see why this claim is true, let $N'\subseteq N$ be any $1$-cohesive group of voters.
Our goal is to show that $A_v\cap W^* \ne \emptyset$ for some $v\in N'$.
Since $W$ satisfies JR, there exists $v\in N'$ such that $A_v\cap W \neq \emptyset$.
If $A_v\cap W \ne \{d_x\}$, then $A_v\cap W^* \ne \emptyset$.
Assume therefore that $A_v\cap W = \{d_x\}$.
Similarly, since $W'$ satisfies JR, there exists $v'\in N'$ such that $A_{v'}\cap W' \neq \emptyset$. 
If $A_{v'} \cap W \cap W' \neq \emptyset$, then since $d_x \not \in W'$, it holds that $A_{v'} \cap W^* \neq \emptyset$. 
Assume therefore that $A_{v'} \cap W \cap W' = \emptyset$. 
Since $e_1,\dots,e_{y-1}\in W$, we have $\emptyset\neq A_{v'}\cap W\subseteq \{e_y,\dots,e_k\}$.
If $e_y\in A_{v'}\cap W$, then $A_{v'}\cap W^* \ne \emptyset$ and we are done.
Else, $A_{v'}\cap W\subseteq \{e_{y+1},\dots,e_k\}$; let $e_z\in A_{v'}\cap W$, and note that $d_x,e_y,e_z$ are arranged in this order.
Since $v$ and $v'$ belong to the $1$-cohesive group $N'$, they share an approved candidate.
If this candidate $c$ comes before $e_y$ in the ordering, then $v'$ approves both $e_z$ and $c$, so by CI, $v'$ also approves $e_y$.
Likewise, if $c$ comes after $e_y$ in the ordering, then $v$ approves both $d_x$ and $c$, so by CI, $v$ also approves $e_y$.
Hence, in both cases, some voter in $N'$ approves the candidate $e_y\in W^*$, establishing the claim.

We now proceed to prove the theorem.
Let $W$ and $W'$ be two JR committees.
Using the claim, we can either replace a candidate from $W\setminus W'$ with one from $W'\setminus W$ or vice versa in such a way that JR is maintained.
By repeating this process, we can construct a path between $W$ and $W'$ of length $d(W,W')$.
\end{proof}

Next, we consider the VI domain.
An instance is in the VI domain if there exists an ordering $v_1, \dots, v_n$ of the voters such that for each candidate $c$, the support set $N_c$ forms an interval with respect to the ordering---more formally, $\{i \in [n] \mid c \in A_i\} = \{i \in [n] \mid i^* \le i\le j^*\}$ for some $i^*, j^*$. 
We say that a candidate $c$ is \emph{Pareto-dominated} by a candidate $c'$ if $N_c$ is a strict subset of $N_{c'}$. 
A candidate is \emph{Pareto-optimal} if it is not Pareto-dominated by any other candidate.

\begin{restatable}{theorem}{corvi}
On the VI domain, the set of JR committees is connected. 
Moreover, if two JR committees $W,W'$ contain no Pareto-dominated candidates, then they can be connected by a path of JR committees of length $d(W,W')$.
\end{restatable}
\begin{proof}
    First, observe that if for two committees $W_1$ and $W_2$ it holds that $N_{W_1} \subseteq N_{W_2}$ and $W_1$ satisfy JR, then $W_2$ also satisfies JR. 
    Therefore, given a committee that satisfies JR, if we replace a Pareto-dominated candidate $c$ with a candidate that Pareto-dominates $c$, the resulting committee still satisfies JR. 
    The same is true if we replace a Pareto-dominated candidate $c$ with an arbitrary candidate, provided that there is a candidate that Pareto-dominates $c$ in the committee.
    We perform a case distinction based on the number of candidates that are Pareto-optimal.
    
    \textbf{Case 1:}  There are at least $k$ Pareto-optimal candidates in the instance.
    We apply the following procedure for $W$ (and analogously for $W'$). 
    For every candidate $c \in W$ that is Pareto-dominated, we replace $c$ with a Pareto-optimal candidate $c' \not \in  W$ that Pareto-dominates $c$. 
    If every candidate that Pareto-dominates $c$ is either in $W$ or also Pareto-dominated (by some candidate in $W$), we replace $c$ with any Pareto-optimal candidate from $C \setminus W$; there always exists such a candidate since we assume that there are at least $k$ Pareto-optimal candidates. 
    During this procedure, both $W$ and $W'$ still satisfy JR due to the observations above.

    After applying this procedure, remove all Pareto-dominated candidates from the instance.
    Observe that if a $1$-cohesive group witnesses a JR violation together with a Pareto-dominated candidate $c$, it also witnesses the violation together with a candidate $c'$ that Pareto-dominates $c$. Therefore, each committee that satisfies JR in the instance without Pareto-dominated candidates also satisfies JR in the instance with those candidates; the converse also holds if the committee does not contain Pareto-dominated candidates.
    \citet[][Proposition 4.3]{DBW+24a} have shown that a VI instance without Pareto-dominated candidates is a CI instance. 
    Therefore, we can apply Theorem \ref{thm:JR_CI_direct_connection} to conclude that $W$ and $W'$ are connected.
    Moreover, if $W$ and $W'$ contain no Pareto-dominated candidates to begin with, the procedure in the previous paragraph is vacuous, and so $W$ and $W'$ are connected by a path of length $d(W,W')$.

    \textbf{Case 2:}
    There are fewer than $k$ Pareto-optimal candidates in the instance.
    Note that in this case, $W$ and $W'$ must contain Pareto-dominated candidates.
    We apply the same procedure as in Case~1, except that if there does not exist another Pareto-optimal candidate that we can replace a Pareto-dominated candidate $c\in W$ with, we terminate the procedure. 
    After this procedure, both $W$ and $W'$ are transformed into committees $W_O$ and $W'_O$ that contain all Pareto-optimal candidates in the instance. 
    Trivially, every committee that contains a subcommittee with all Pareto-optimal candidates in the instance satisfies JR. 
    Therefore, one by one, we can replace all Pareto-dominated candidates in $W_O$ with the Pareto-dominated candidates in $W'_O$, which means that $W_O$ and $W'_O$ are connected.    
\end{proof}

By contrast, shortest-path connections may not be possible if the committees contain Pareto-dominated candidates.

\begin{restatable}{proposition}{ExampleJRNotDirectlyConnectedOnVI}
    On the VI domain, there exist JR committees $W,W'$ that cannot be connected by a path of JR committees of length $d(W,W')$.
        \label{prop:VI-non-directly}
\end{restatable}
\begin{proof}
    Consider the instance with voters $N = \{v_1, \dots, v_6\}$, candidates $C = \{c_1,\dots,c_6\}$, and $k = 2$ (so $\frac{n}{k} = 3$). Furthermore, let $c_1 \in A_v$ for $v \in \{v_1,v_2\}$, $c_2 \in A_v$ for $v \in \{v_5,v_6\}$,
    $c_3 \in A_v$ for $v \in \{v_3\}$, 
    $c_4 \in A_v$ for $v \in \{v_4\}$,
    $c_5 \in A_v$ for $v \in \{v_1, \dots, v_4\}$, and
    $c_6 \in A_v$ for $v \in \{v_3, \dots,v_6\}$. A visualization of this profile can be seen in \Cref{tab:VI-notdirectly}.
    Clearly, this profile is in VI.
    Since only $c_5,c_6$ are approved by at least  $\frac{n}{k} = 3$ voters, for a committee $W_{\mathrm{JR}}$ to satisfy JR, it is sufficient that at least two voters in $\{v_1,v_2,v_3,v_4\}$ and at least two voters in $\{v_3,v_4,v_5,v_6\}$ approve at least one candidate in $W_{\mathrm{JR}}$. 
    Therefore,  both $W = \{c_1,c_2\}$ and $W' = \{c_3,c_4\}$ satisfy JR.
    On the other hand $W_1 = \{c_1,c_3\}$ and $W_2 = \{c_1,c_4\}$ do not satisfy JR---the witnesses for a JR violation are $\{v_4,v_5,v_6\}$ and $c_6$ for $W_1$, and $\{v_3,v_5,v_6\}$ and $c_6$ for $W_2$. 
    Similarly,  $W_3 = \{c_2,c_3\}$ and $W_4 = \{c_2,c_4\}$ do not satisfy JR.
    Therefore, $W$ and $W'$ cannot be connected by a path of JR committees of length $d(W,W') = 2$.
\end{proof}

\begin{table}[tbh]
    \centering
    \begin{tabular}{@{}lcccccc@{}}
    \toprule
      & $v_1$ & $v_2$ & $v_3$ & $v_4$ & $v_5$ & $v_6$ \\ \midrule
    $c_1$ & $\times$ & $\times$ &  &        &        &        \\
    $c_2$ &  &  &  &        &      $\times$  & $\times$       \\
    $c_3$ &  &  & $\times$ &  &  &  \\
    $c_4$ &  &  &  & $\times$ &  &  \\
    $c_5$ & $\times$ & $\times$ & $\times$ & $\times$ &  &  \\
    $c_6$ &        &        & $\times$       & $\times$ & $\times$ & $\times$ \\
\bottomrule
\end{tabular}
\caption{The profile constructed for the proof of \Cref{prop:VI-non-directly}. A mark in entry $(i,j)$ indicates that $c_i$ is approved by $v_j$.}\label{tab:VI-notdirectly}
\end{table}

For EJR, we show that the set of committees satisfying the notion may not be connected through shortest paths, even for instances that simultaneously satisfy VI and CI.

\begin{restatable}{proposition}{EJRNotDirectlyConnectedOnVICI}
\label{EJR+_Not_Direct_On_VI_Cap_CI}
On the CI and VI domains, there exist EJR committees $W,W'$ that cannot be connected by a path of EJR committees of length $d(W,W')$.
\end{restatable}
\begin{proof}
Consider the instance with voters $N = \{v_1, \dots, v_6\}$, candidates $C = \{c_1, \dots, c_7\}$, $k = 3$ (so $\frac nk = 2$), and the following approval ballots: 
    $A_1 = A_2 = A_3 = \{c_1,c_2, c_3,c_4,c_5\}$ and $A_4 = A_5 = A_6 = \{c_3,c_4, c_5,c_6,c_7\}$.
A visualization of this profile can be seen in \Cref{tab:EJR+_Not_Direct_On_VI_Cap_CI}.    
It is easy to see that this instance is in both CI and VI.

First, we claim that the committee $\{c_1,c_2,c_3\}$ satisfies EJR. For this, observe that the first three voters are fully satisfied, while the last three voters can only claim a single candidate. 
    As the latter voters approve $c_3$, they do not constitute an EJR violation. 
    Due to the symmetry of the instance, the committee $\{c_3, c_6, c_7\}$ also satisfies EJR. 
    Now, note that EJR in this instance requires that some voter to approve all three candidates in the committee, since all voters together form a $3$-cohesive group. 
    However, after replacing $c_1$ or $c_2$ with $c_6$ or $c_7$ in $\{c_1,c_2,c_3\}$, no voter approves three candidates in the committee.
\end{proof}

\begin{table}[tbh]
     \centering
     \begin{tabular}{@{}lcccccc@{}}
     \toprule
       & $v_1$ & $v_2$ & $v_3$ & $v_4$ & $v_5$ & $v_6$ \\ \midrule
     $c_1$ & $\times$ & $\times$ & $\times$ &        &        &        \\
     $c_2$ & $\times$ & $\times$ & $\times$ &        &        &        \\
     $c_3$ & $\times$ & $\times$ & $\times$ & $\times$ & $\times$ & $\times$ \\
     $c_4$ & $\times$ & $\times$ & $\times$ & $\times$ & $\times$ & $\times$ \\
     $c_5$ & $\times$ & $\times$ & $\times$ & $\times$ & $\times$ & $\times$ \\
     $c_6$ &        &        &        & $\times$ & $\times$ & $\times$ \\
     $c_7$ &        &        &        & $\times$ & $\times$ & $\times$ \\ \bottomrule
 \end{tabular}
 \caption{The profile constructed for the proof of \Cref{EJR+_Not_Direct_On_VI_Cap_CI}. A mark in entry $(i,j)$ indicates that $c_i$ is approved by $v_j$.}
 \label{tab:EJR+_Not_Direct_On_VI_Cap_CI}
 \end{table}

This leads us to the following open question.
\begin{openq}
    Is the set of EJR committees connected on the CI or VI domain?
\end{openq}

\section{Conclusion and Future Directions}

We have studied the problem of reconfiguring proportional committees. 
In particular, we ask whether, for any two given proportional committees, there is a reconfiguration path consisting only of proportional committees such that each transition corresponds to swapping a pair of candidates. 
We demonstrate that committees satisfying justified representation (JR) or extended justified representation (EJR) do not always exhibit this type of connectivity.
In spite of this, we obtain positive results in three different directions.
Firstly, any two JR committees can be connected via a transition path consisting only of $2$-JR committees, while any two EJR committees can be connected via $4$-EJR committees. 
Secondly, the negative result for JR does not apply to several popular voting rules: the committees produced by these rules belong to the same connected component within the set of JR committees. 
Thirdly, for two important restricted domains, the set of JR committees is guaranteed to be connected. 

Besides the open questions already mentioned, our work leaves several intriguing directions for future research. 
For example, the connectedness of EJR committees is still far less understood than that of JR committees.
In particular, it remains open whether it is always possible to connect two EJR committees via a path consisting only of $2$-EJR committees, or whether the set of EJR committees is guaranteed to be connected in restricted domains.
Moreover, our result that several well-known voting rules produce committees that can be connected by JR committees raises the following high-level question: 
All of these voting rules appear to select stronger committees than necessitated by the proportionality axioms. 
Indeed, the committee in the proof of \Cref{thm:isolated_jr} intuitively contains weak candidates and is not selected by common voting rules. 
Despite the perceived lack of quality, this committee satisfies all known proportionality axioms. 
Is there a principled way to explain this phenomenon, and what are good methods for distinguishing between ``strong'' and ``weak'' proportional committees? 

\section*{Acknowledgments}

This work was partially supported by the Deutsche Forschungsgemeinschaft under grants BR 2312/11-2 and BR 2312/12-1, by the Singapore Ministry of Education under grant number MOE-T2EP20221-0001, and by an NUS Start-up Grant.

\bibliographystyle{named}
\bibliography{group, algo, temp}
\clearpage
\newpage 
\newpage

\end{document}